\documentclass[12pt]{article}
\usepackage{
amsmath,amssymb,mathrsfs,bbm,a4wide}
\usepackage{tikz}
\usetikzlibrary{matrix,arrows}

\newenvironment{proofof}[1]{%
\par\addvspace{12pt plus3pt minus3pt}\global\logotrue%
\noindent{\bf Proof of #1.\hskip.5em}\ignorespaces}{%
	\par\iflogo\prbox\par
	\addvspace{12pt plus3pt minus3pt}\fi}

\newif\iflogo
\def\prbox{\par
	\vskip-\lastskip\vskip-\baselineskip\hbox to \hsize{\hfill\fboxsep0pt\fbox{\phantom{\vrule width5pt height5pt depth0pt}}}\global\logofalse}

\newenvironment{proof}{%
\par\addvspace{12pt plus3pt minus3pt}\global\logotrue%
\noindent{\bf Proof.\hskip.5em}\ignorespaces}{%
	\par\iflogo\prbox\par
	\addvspace{12pt plus3pt minus3pt}\fi}

\newcommand{\II}{\mathbbm{1}}

\newtheorem{thm}{Theorem}[section]
\newtheorem{dfn}[thm]{Definition}
\newtheorem{lem}[thm]{Lemma}
\newtheorem{prp}[thm]{Proposition}
\newtheorem{cor}[thm]{Corollary}

\newcommand{\balpha}{{\boldsymbol{\alpha}}}
\newcommand{\bbeta}{{\boldsymbol{\beta}}}
\newcommand{\bgamma}{{\boldsymbol{\gamma}}}
\newcommand{\bdelta}{{\boldsymbol{\delta}}}
\newcommand{\bzeta}{{\boldsymbol{\zeta}}}
\newcommand{\bchi}{{\boldsymbol{\chi}}}
\newcommand{\bxi}{{\boldsymbol{\xi}}}
\newcommand{\bpi}{{\boldsymbol{\pi}}}
\newcommand{\bPi}{{\boldsymbol{\Pi}}}
\newcommand{\bvarpi}{{\boldsymbol{\varpi}}}

\newcommand{\bC}{{\boldsymbol{C}}}
\newcommand{\bE}{{\boldsymbol{E}}}
\newcommand{\be}{{\boldsymbol{e}}}
\newcommand{\bfab}{{\boldsymbol{f}}}
\newcommand{\bg}{{\boldsymbol{g}}}
\newcommand{\bh}{{\boldsymbol{h}}}
\newcommand{\bk}{{\boldsymbol{k}}}
\newcommand{\bn}{{\boldsymbol{n}}}
\newcommand{\bp}{{\boldsymbol{p}}}
\newcommand{\bu}{{\boldsymbol{u}}}
\newcommand{\bv}{{\boldsymbol{v}}}
\newcommand{\bw}{{\boldsymbol{w}}}
\newcommand{\bX}{{\boldsymbol{X}}}
\newcommand{\bV}{{\boldsymbol{V}}}
\newcommand{\bN}{{\boldsymbol{N}}}

\DeclareMathOperator{\Data}{Data}
\DeclareMathOperator{\supp}{supp}
\DeclareMathOperator{\range}{range}

\newcounter{tightenum}
\newenvironment{tightenumerate}%
{\begin{list}{(\roman{tightenum})}{\usecounter{tightenum} \setlength{\itemsep}{0pt}\setlength{\parsep}{0pt}\setlength{\topsep}{0pt}}}%
{\end{list}}

\numberwithin{equation}{section}

\begin{document}

\title{Quantization of linearized gravity in cosmological vacuum spacetimes}
\author{Christopher J Fewster\thanks{\tt chris.fewster@york.ac.uk}~ and
David S Hunt\thanks{\tt dsh501@york.ac.uk}\\ Department of Mathematics,
                 University of York, \\
                 Heslington,
                 York YO10 5DD, U.K.}
\maketitle

\maketitle

\begin{abstract}
Linearized Einstein gravity (with possibly nonzero cosmological constant) is quantized in the framework of algebraic quantum field theory
by analogy with Dimock's treatment of electromagnetism [\emph{Rev.\ Math.\
  Phys.} \textbf{4} (1992) 223--233]. To achieve this, the classical
theory is developed in a full, rigorous and systematic fashion, with particular attention
given to the circumstances under which the symplectic product is weakly non-degenerate
and to the related question of whether the space of solutions is 
separated by the classical observables on which the quantum theory is modelled.
\end{abstract}

\section{Introduction}

Low energy effects of quantum gravitation can usefully be modelled by the quantum field theory
of linearized gravity, treated as a quantum field theory in a fixed background spacetime. The
background of de Sitter space is particularly interesting from the cosmological viewpoint, where
linearized gravitons can induce fluctuations in the cosmic microwave background, for instance, see \cite[Ch.~4 \&~10]{Weinberg} for a discussion of inflation and tensor fluctuations.
For these reasons, the quantum field theory of linear gravitational perturbations has been studied extensively on several spacetimes, see for example \cite{FordParker} and references therein for the case of Robertson-Walker spacetimes.

In the important de Sitter case, there has been a long-running controversy concerning the existence of
a suitable de Sitter invariant vacuum state. This arises because natural constructions, for example, in the 
transverse-traceless and synchronous gauge associated with conformally flat coordinates on the
Poincar\'e patch~\cite{FordParkerInfrared} lead to a divergent expression for the two-point function as a result
of an infrared divergence in the integral over modes.  Some parties to the controversy~\cite{HiguchiMarolf} nonetheless claim the existence of a de Sitter invariant two-point function whilst others \cite{MiaoWoodward} deny those claims. At
issue is the validity of the use of Euclidean and analytic continuation methods, and the freedom to
select gauge conditions or add gauge-fixing terms. 

For this reason it seems appropriate to have a framework for the quantization of linear gravitational perturbations on a general cosmological vacuum spacetime that will, when required, allow one to reduce down to a specific choice of spacetime and will also permit a rigorous investigation to be made into the issue of Hadamard states. This paper is the first step in that process.

Of course the quantization of linear gravitational perturbations has previously been considered by numerous authors and it would not be feasible to list them all here, but, for our purposes~\cite{Ottewill} in particular stands out because of its consideration of Hadamard renormalization of the graviton stress tensor on cosmological vacuum spacetimes. However, this was achieved by introducing gauge-breaking terms and ghost fields. We prefer to avoid the introduction of such auxiliary fields and choices of gauge-breaking terms and apply a minimal approach along the same lines as Dimock's
quantization of electromagnetism~\cite{DimockEM} (see also \cite{CJFMJP}). A brief outline of the method is as follows. 
Let $(M,\bg)$ be a background spacetime that solves the cosmological vacuum Einstein equation.
The classical phase space $\mathscr{P}(M)$ of linearized gravity consists of solutions to the 
linearized Einstein equation modulo gauge-equivalence with a typical equivalence class 
denoted $[\bgamma]$. A classical observable is a function from $\mathscr{P}(M)$
to the set of complex numbers. For example, given any $[\bgamma]\in \mathscr{P}(M)$, we
may smear any representative $\bgamma$ (i.e., a solution to the linearized Einstein equation) 
against a  compactly supported test-tensor $\bfab$, to
give an observable
\[
F_{\bfab}([\bgamma]) = \int_M \gamma_{ab}f^{ab} dvol_\bg,\;
\]
which is independent of the choice of representative provided that $\bfab$ satisfies the condition $\nabla^{a}f_{(ab)}=0$. Once a presymplectic form is given on $\mathscr{P}(M)$, we
may quantize to form a $*$-algebra of observables -- essentially by Dirac quantization. 
This is a particular instance of the algebraic approach to quantum field theory in curved spacetimes
(see, e.g., \cite[Ch.~4.5]{WaldQFT} for an introduction). The advantages of this procedure are: (a) it is manifestly independent of any choice of gauge, (b)
it does not involve the addition of gauge-breaking terms or auxiliary fields in the action, (c) the
method can be implemented in arbitrary globally hyperbolic cosmological vacuum background spacetimes, (d) it separates the construction of the algebra of observables from questions concerning the 
existence or otherwise of particular vacuum states which can then be addressed separately,
and (e) it circumvents the known nonexistence of a Wightman theory of linearized gravity on
Minkowski space  allowing arbitrary smearings \cite{StrocchiII,BracciStrocchi}.
Nonetheless, there is a disadvantage, namely that we are not free to smear the metric perturbation
against arbitrary tensors, which would be necessary in order to couple it to other fields. However,
this quantization provides a rigorous framework that could be broadened in further work, 
for example, by exploiting the recent and rigorous approach to the Batalin--Vilkovisky formulation 
due to Fredenhagen and Rejzner~\cite{FredenhagenRejzner}. Following the work reported here, Hack and Schenkel~\cite{HackSchenkel2012} have shown that our treatment of linearized gravity can be put into a broader framework of quantisations of linear theories with gauge invariance. It would also be interesting
to adapt these methods to other theories related to linearized gravity~\cite{DeserWaldronII,DeserWaldronI,DeserWaldronIII}.

An important question is whether or not the class of observables identified above is sufficient
to distinguish different equivalence classes of solutions. This is closely related to the issue 
of whether the pre-symplectic form is weakly non-degenerate, and thus symplectic. We
will be able to resolve these questions positively at least in the case where the background
spacetime has compact Cauchy surfaces, by adapting splitting results of Moncrief~\cite{Moncriefdecomp} derived
in the Arnowitt-Deser-Misner (ADM) formalism (briefly reviewed in Appendix~\ref{sec:ADMnondegen}).

As mentioned above, the quantized theory leans heavily on the theory of classical linearized gravity
and a substantial part of the present paper is devoted to a clear and general presentation of the theory. Much of this material is known, of course, but we have not found a unified and full treatment in the literature that would be sufficient for our purposes.\footnote{For instance \cite[Sec.~4]{FM}
considers only compact Cauchy surfaces, while~\cite{Lichnerowicz} restricts to local results.}  
We therefore hope that our presentation of the
theory may be of independent interest and utility. To a large extent we emphasize a four-dimensional
`spacetime' viewpoint on the theory in contrast to the `sliced' viewpoint of the ADM formulation. 
This is more natural for the formulation of the quantum theory and also removes any suspicion
of dependence on particular slicings, coordinates or choices of (linearized) lapse and shift. 
However, as mentioned, we will draw on insights from the ADM formulation and make contact where necessary. 

The paper is structured as follows. In section~\ref{sec:lingrav} we begin with a brief introduction to linearized gravity and the gauge invariance of the theory. While our approach as a whole is
independent of gauge choice, particular gauges are used for technical purposes to establish
results on the full solution space and, to this end, we describe three gauges, de Donder, transverse-traceless and synchronous, and the circumstances under which they can be employed. This is largely standard, but
for our discussion of the transverse-traceless gauge, for which we find that the topology of the Cauchy surface determines whether one may pass globally to the transverse-traceless gauge. In
the synchronous case it is shown that a general perturbation is gauge-equivalent to a synchronous
perturbation in a normal neighbourhood of any Cauchy surface (relative to the induced normal field).

In section~\ref{sec:ExistenceUniqueness} we review existence and uniqueness of solutions to the linearized Einstein equation and provide the proofs for perturbations that are spacelike-compact, i.e., 
supported in the union of the causal future and past of a compact set. Similar results are sketched in \cite{FM}  for the case of compact Cauchy surfaces. We also prove various results concerning Green's operators, which will be required in section~\ref{sec:phasespace}. The 
upshot of this discussion is that, modulo gauge equivalence, all spacelike-compact solutions
to linearized gravity are given by the action of the advanced-minus-retarded solution, associated to the Lichnerowicz Laplacian $P_{ab}^{\phantom{ab}cd}=\nabla^{e}\nabla_{e}\delta^{c}_{\phantom{c}a}\delta^{d}_{\phantom{d}b}-2R^{c \phantom{ab}d}_{\phantom{c}ab}$,
acting on smooth, compactly supported, symmetric tensor fields $\bfab$ obeying
$\nabla^a (f_{ab}-\frac{1}{2} f_c^{\phantom{c}c} g_{ab}) = 0$. One should note that although Lichnerowicz considers various Green's operators in \cite{Lichnerowicz}, his existence 
results are purely local in nature, that is, they are valid in suitable open subsets of the manifold. Likewise, the book by Friedlander~\cite{Friedlander} also only considers local results. Our approach rests on
the global theory of normally hyperbolic equations (developing the results just mentioned) that is described in the book of B\"{a}r, Ginoux and Pf\"{a}ffle~\cite{Bar}. We make contact with Lichnerowicz's work by
confirming that our propagator is identical to the one posited in~\cite{Lichnerowicz}. He justified its form by using analogy with electromagnetism and the work of Fierz and Pauli~\cite{FierzPauli} for the case that the background spacetime was Minkowski. In the present paper, it emerges (in any
cosmological vacuum spacetime) as the Dirac 
quantization of Poisson brackets of the classical observables $F_{\bfab}$. 

The final part is section~\ref{sec:phasespace}, which deals with construction of the phase space for linearized gravity and its subsequent quantization. As mentioned, the weak non-degeneracy of the symplectic product is established for compact Cauchy surfaces using results from the ADM formalism, in particular results of \cite{Moncriefdecomp} providing
various splittings of the space of initial data. Although these results are not established on 
general non-compact Cauchy surfaces, we conjecture that the non-degeneracy (as we formulate it) will hold for a large class of spacetimes with non-compact Cauchy surfaces as well. This is analogous to the situation in electromagnetism \cite{DimockEM} where to prove non-degeneracy one appeals to the Hodge decomposition, again placing restrictions on the Cauchy surface. 

Once the phase space is constructed, the quantization then follows standard lines, either by Dirac quantization of the observables or a direct quantization of the symplectic space; we demonstrate that they are equivalent in section~\ref{sec:phasespacereform}. We prove that the algebra satisfies the time-slice condition, i.e., it is generated by smeared fields whose smearing tensors are supported in
any neighbourhood of a given Cauchy surface. We also briefly discuss the extent to which the construction
respects the axioms of local covariance~\cite{BFV}.

Summarizing, our approach puts the quantum field theory of linearized gravity on a rigorous mathematical footing, which will permit future investigations into the states of the theory. In particular, it will enable precise statements to be made, using techniques from microlocal analysis, about Hadamard states
(cf.\ \cite{CJFMJP} for the case of electromagnetism).  Elsewhere, we  will investigate what
can be said concerning the existence of a de Sitter invariant Hadamard state in the context of our framework. Note that the discussion of linearization instabilities would manifest itself in the properties of states~\cite{MoncriefQuantInstab} of the algebra of observables we construct and so we do not consider these issues here.

We end this introduction with some preliminaries. We take, as in \cite{HawkingEllis,Wald}, the Riemann tensor to be defined by 
\begin{equation}
\label{eqn:Riemann}
R_{abc}^{\hspace{0.4cm} d} \omega_{d} = (\nabla_{a} \nabla_{b} - \nabla_{b} \nabla_{a}) \omega_{c}
\end{equation}
and the Ricci scalar to be
\begin{equation}
\label{eqn:Ricci}
R_{ac} = R_{abc}^{\hspace{0.4cm} b}.
\end{equation}
Boldface type is used to indicate a tensor written with its indices suppressed. We also use $\bw^{\flat}$ to denote the covariant form of a vector field $\bw$ and $\bv^{\sharp}$ to denote the contravariant form of a covector field $\bv$.
A number of spaces of smooth scalar and tensor fields will be employed; all
are taken to be complex-valued unless otherwise stated. As usual, $C^{\infty}(M)$  denotes the space of smooth functions on $M$, while we write $C^{\infty}(T^{a}_{b}(M))$ for the space of smooth rank $(a,b)$ tensor fields on $M$.
Various subscripts will denote restrictions on the support of such functions; 
$C^{\infty}_{0}(M)$ and $C^{\infty}_{0}(T^{a}_{b}(M))$ denote the
compactly supported elements of  $C^{\infty}(M)$  and $C^{\infty}(T^{a}_{b}(M))$, while, following the notation of~\cite[p.~90]{Bar}, the subscript $SC$ (spacelike-compact) attached to any of these spaces (e.g., 
$C^{\infty}_{SC}(T^{a}_{b}(M))$) denotes the subspace of tensor fields of the appropriate type whose support is contained within $J(K)$ for some compact $K \subset M$. In particular this means that the intersection of the support with a Cauchy surface is compact.  The subscript $TC$ (timelike-compact) denotes
tensor fields whose support lies between two Cauchy surfaces. Finally,  $S^{0}_{2}(M)$ (resp., 
$S^{2}_{0}(M)$) denotes the symmetric elements of $T^{0}_{2}(M)$ (resp., $T^{2}_{0}(M)$), giving rise to corresponding spaces of symmetric tensor fields.

\section{Linearized gravity}
\label{sec:lingrav}

A spacetime is a pair $(M,\bg)$ consisting of a four-dimensional, smooth, real, connected, Hausdorff, orientable manifold without boundary $M$ together with a smooth Lorentzian metric $\bg$ of signature $(-+++)$, with respect
to which $M$ is time-orientable; such a spacetime is automatically paracompact by the theorem in the appendix of \cite{Geroch}. Each (background) spacetime will be assumed to be globally hyperbolic, i.e., it admits no closed causal curves and for any two points $p,q \in M$ the set $J^{+}(p) \cap J^{-}(q)$ is compact~\cite{BernalSanchezcausal}. In addition, we assume that the metric solves
the vacuum Einstein equations with cosmological constant
\begin{equation}
\label{eqn:Einsteineqn}
G_{ab}+\Lambda g_{ab}=0,
\end{equation}
where $G_{ab} = R_{ab}-\frac{1}{2}Rg_{ab}$ as usual.
As a consequence of \eqref{eqn:Einsteineqn}, the background Ricci tensor and Ricci scalar obey
\begin{equation}
\label{eqn:Riccitenscalvalue}
R_{ab} = \Lambda g_{ab},\qquad R = 4\Lambda,
\end{equation}
which will often be used in what follows. 

Our aim is to quantize linearized perturbations of the Einstein equations. To linearize, we consider a  smooth one-parameter family of solutions $\lambda \mapsto \bg(\lambda)$ to \eqref{eqn:Einsteineqn} with $\bg(0)=\bg$. 
The linearized equation satisfied by the perturbation $\bgamma:=\dot{\bg}(0)
\in C^\infty(T^0_2(M))$ on a cosmological background is
\begin{equation}
\label{eqn:linearizedequations}
L_{ab}(\bgamma)=0,
\end{equation}
where
\begin{align}
\label{eqn:nonhypeqn}
L_{ab}(\bgamma)& = -\frac{1}{2} g_{ab} (\nabla^{c} \nabla^{d} \gamma_{(cd)} - \Box \gamma - \Lambda \gamma) - \Lambda \gamma_{(ab)} - \Box \gamma_{(ab)} \notag\\
&\qquad\qquad - \frac{1}{2} \nabla_{a} \nabla_{b} \gamma + \frac{3}{2}\nabla^{c} \nabla_{(a} \gamma_{bc)},
\end{align}
$\gamma=\gamma^a_{\phantom{a}a}$ is the trace of $\bgamma$ and 
$\Box = g^{ab}\nabla_{a}\nabla_{b}$. The linearized Einstein tensor $L_{ab}(\bgamma)$ vanishes for antisymmetric $\bgamma$ and for symmetric $\bgamma$ (our case of interest) reduces to
\begin{equation}
\label{eqn:symmL}
L_{ab}(\bgamma)= -\frac{1}{2} g_{ab} (\nabla^{c} \nabla^{d} \gamma_{cd} - \Box \gamma - \Lambda \gamma) - \Lambda \gamma_{ab} \\ - \frac{1}{2} \Box \gamma_{ab} - \frac{1}{2} \nabla_{a} \nabla_{b} \gamma + \nabla^{c} \nabla_{(a} \gamma_{b)c}.
\end{equation}
The linearized equation \eqref{eqn:symmL} also follows from the Euler-Lagrange equations of the Lagrangian (with $\bgamma$ assumed to be symmetric)
\begin{equation}
\label{eqn:lagrangian}
\mathscr{L}=T^{abcdef}\nabla_{a}\gamma_{bc}\nabla_{d}\gamma_{ef}+S^{abcd}\gamma_{ab}\gamma_{cd},
\end{equation}
where $T^{abcdef}$ is given by
\begin{align}
T^{abcdef}&=\frac{1}{4}(g^{ad}g^{bc}g^{ef}+g^{af}g^{d(b}g^{c)e}+g^{d(b}g^{c)f}g^{ae} 
\notag\\
& \qquad\qquad -g^{ad}g^{e(b}g^{c)f}-g^{a(e}g^{f)d}g^{bc}-g^{d(b}g^{c)a}g^{ef}),
\end{align}
which is symmetric on interchange of $b$ to $c$ and on interchange of $e$ to $f$ as required by the symmetry of $\bgamma$; $T^{abcdef}$ is also symmetric on interchange of $abc$ to $def$. Finally 
\begin{equation}
S^{abcd}=\frac{\Lambda}{4}g^{ac}g^{bd}+\frac{\Lambda}{4}g^{bc}g^{ad}-\frac{\Lambda}{4}g^{ab}g^{cd}
\end{equation}
is symmetric on interchange of $ab$ to $cd$ and symmetric on interchange of $a$ to $b$ and interchange of $c$ to $d$. The Lagrangian \eqref{eqn:lagrangian} comes from the second order expansion of the Einstein-Hilbert action with cosmological constant
\begin{equation}
S=\int{(R-2\Lambda) dvol_{\bg}},
\end{equation}
where $dvol_{\bg}$ denotes the volume element on spacetime associated with $\bg$. Note that in the expansion, the linear term is a total divergence and the zeroth-order term is the Einstein-Hilbert action for the background.

The covariant conjugate momentum
\begin{equation}
\label{eqn:momentum}
\Pi^{abc} = \frac{1}{\sqrt{-g}}\frac{\delta S}{\delta \nabla_{a}\gamma_{bc}}  = 2T^{abcdef}\nabla_{d}\gamma_{ef}
\end{equation}
is given by
\begin{equation}
\label{eqn:PIabc}
\Pi^{abc} = -\frac{1}{2}\nabla^{a}\gamma^{bc}+\frac{1}{2}g^{bc}\nabla^{a}\gamma-\frac{1}{2}g^{bc}\nabla_{d}\gamma^{ad} - \frac{1}{4}g^{ac}\nabla^{b}\gamma- \frac{1}{4}g^{ab}\nabla^{c}\gamma+\frac{1}{2}\nabla^{b}\gamma^{ac}+\frac{1}{2}\nabla^{c}\gamma^{ab}
\end{equation}
and the Euler-Lagrange equations are thus
\begin{equation}
\label{eqn:Eulerlagrange}
\nabla_{c}\Pi^{cab}-2S^{abcd}\gamma_{cd}=L^{ab}(\bgamma)=0.
\end{equation}
Note that the first equality holds for all smooth symmetric $\bgamma$, not just for solutions to
the linearized equations.

The linearized Einstein equation \eqref{eqn:linearizedequations} is non-hyperbolic. To address existence and uniqueness properties of solutions we exploit gauge freedom to put the equation into a hyperbolic form. As is well known, the gauge transformations for linearized gravity take the form
\begin{equation}\label{eqn:gaugetrans}
\bgamma'=\bgamma+ \pounds_\bw\bg, \qquad\text{i.e.,}\quad \gamma^{\prime}_{ab}=\gamma_{ab}+2\nabla_{(a}w_{b)}
\end{equation}
for $\bw\in C^\infty(T^1_0(M))$. Owing to the identity 
\begin{equation}
\label{eqn:SWliederiv}
L_{ab}(\pounds_{\bw}\bg)=\pounds_{\bw}(G_{ab}+\Lambda g_{ab}),
\end{equation}
the linearized equation \eqref{eqn:linearizedequations} is gauge invariant, in the sense that
$L_{ab}(\pounds_{\bw}\bg)=0$,
if and only if the background is a cosmological vacuum solution (see \cite[Lem.~2.2]{Stewart}). This restricts our analysis to such spacetimes. 

In addition, to ensure well-posedness of the Cauchy problem and the existence of certain integrals, we restrict attention to spacelike-compact perturbations. With this in mind we introduce the notations
\begin{align}
\mathscr{T}(M) &=  C_{SC}^{\infty}(S^{0}_{2}(M))  \\
\mathscr{S}(M) & = \{ \bgamma \in \mathscr{T}(M) \hspace{0.2cm} | \hspace{0.2cm} L_{ab}(\bgamma)=0\} \\
\mathscr{G}(M) &= \{ \pounds_{\bw}\bg \hspace{0.2cm} | \hspace{0.2cm} \bw \in C_{SC}^{\infty}(T^{1}_{0}(M))  \}
\end{align}
for, respectively, the spacelike-compact symmetric tensors,
the subspace obeying the linearized Einstein equation and the pure gauge solutions induced by spacelike-compact vector fields. 

In section~\ref{sec:phasespace}
 it will be necessary to consider the larger class of spacelike-compact pure gauge
solutions induced by arbitrary smooth vector fields,
\begin{equation}\label{eq:Ghatdef}
\hat{\mathscr{G}}(M)=\{\pounds_\bw\bg: \bw\in C^\infty(T^1_0(M))\}
\cap \mathscr{T}(M).
\end{equation}
In fact, this coincides with $\mathscr{G}(M)$ whenever $M$ has
compact Cauchy surfaces; more generally, the two sets differ only if 
there exist regions of the form $M\setminus J(K)$,
for $K$ compact, supporting Killing vector fields. The example of Minkowski space shows
that the latter condition, while necessary, is not sufficient: if $\pounds_\bw\bg\in\hat{\mathscr{G}}(M)$, then $\bw\in C^\infty(T^1_0(M))$ satisfies the Killing equation outside some set $J(K)$ with $K$ compact;
without loss of generality, $J(K)$ may be assumed  to have connected complement, by expanding
$K$ if necessary. Because Minkowski space is maximally symmetric, there exists a global Killing vector field $\bxi$ such that $\supp (\bw-\bxi)\subset J(K)$; as $\pounds_{\bw-\bxi}\bg=\pounds_\bw\bg$,
we see that $\pounds_\bw\bg\in\mathscr{G}(M)$, so $\hat{\mathscr{G}}(M)=\mathscr{G}(M)$ in this case. To see
how $\hat{\mathscr{G}}(M)$ and $\mathscr{G}(M)$ can differ, consider Minkowski space with the causal future and past of the origin removed, which is still a globally hyperbolic Einstein manifold $M$ and inherits all the Killing vector fields of Minkowski space. Let $r$ be the standard radial coordinate in the time-zero hyperplane and let $K$ be the set of all points in the time-zero hyperplane with $R\le r\le 2R$, for some $R>0$. Then $M\setminus J(K)$ is disconnected, and neither component is relatively compact. Take any Killing vector field $\bxi$ and let $\bw=f(r)\bxi$, where $f$ is constant outside $[R,2R]$ with $f(R)\neq f(2R)$. Then $\pounds_\bw\bg\in
\hat{\mathscr{G}}(M)\setminus\mathscr{G}(M)$. 

The remainder of this section is devoted to identifying when particular gauge choices can be made, and is largely standard with the exception of our remarks on the transverse traceless gauge. We begin with the primary choice of gauge that will be used throughout this paper, the de Donder gauge, in which the linearized Einstein equations  reduce to a hyperbolic form.

\subsection{de Donder gauge}
\label{sec:deDonder}

The de Donder gauge condition for a perturbation $\bgamma$ is $\nabla^{a}\overline{\gamma}_{ab}=0$, where the standard trace-reverse operation (see e.g., \cite[Ch.~7.5]{Wald}) is defined by $\overline{\gamma}_{ab} := \gamma_{ab}-\frac{1}{2}g_{ab} \gamma$
and satisfies $\overline{\gamma}=-\gamma$ and $\overline{\overline{\gamma}}_{ab}=\gamma_{ab}$. Note that for symmetric perturbations there is no ambiguity in writing $\nabla \cdot \bgamma$ for $\nabla^{a}\gamma_{ab}$, and we will often do so.

To any rank $(0,2)$ symmetric tensor, not necessarily a solution to \eqref{eqn:linearizedequations}, there is a gauge equivalent rank $(0,2)$ symmetric tensor that satisfies the de Donder gauge condition. (Actually symmetry is not required in the proof.) Before stating this theorem we consider the following lemma concerning pure gauge perturbations.
\begin{lem} For any $\bw\in C^\infty(T^1_0(M))$, on a cosmological vacuum background spacetime, we have
\label{lem:divtracereversegauge}
\begin{equation}
\nabla \cdot \overline{\pounds_{\bw}\bg} = (\Box+\Lambda)(\bw)^{\flat}.
\end{equation}
\end{lem}
\begin{proof}
By the definition of trace-reversal,  $(\overline{\pounds_{\bw}\bg})_{ab}=\nabla_{a}w_{b}+\nabla_{b}w_{a} - g_{ab}\nabla_{c}w^{c}$. Taking the divergence of this gives
\begin{equation}
\nabla^{a}(\overline{\pounds_{\bw}\bg})_{ab} = \nabla^{a}\nabla_{a}w_{b}+\nabla^{a}\nabla_{b}w_{a}-\nabla_{b}\nabla^{c}w_{c}.
\end{equation}
Using \eqref{eqn:Riccitenscalvalue} we have $\nabla^{a}\nabla_{b}w_{a}=\nabla_{b}\nabla^{a}w_{a}+\Lambda w_{b}$ and hence the desired result. 
\end{proof}

\begin{thm}\label{thm:deDsplit}
The space $\mathscr{T}(M)$ may be decomposed as
\begin{equation}\label{eqn:dDsplit}
\mathscr{T}(M) = \mathscr{T}^{dD}(M) + \mathscr{G}(M),
\end{equation}
where $\mathscr{T}^{dD}(M)= \{ \bgamma \in \mathscr{T}(M) \hspace{0.2cm} | \hspace{0.2cm} \nabla^{a}\overline{\gamma}_{ab}=0\}$. 
The intersection $\mathscr{G}^{dD}(M) = \mathscr{T}^{dD}(M) \cap \mathscr{G}(M)$ is given by
\begin{equation}
\mathscr{G}^{dD}(M) =\{ \pounds_\bw\bg \hspace{0.2cm} | \hspace{0.2cm} 
\bw\in C^\infty_{SC}(T^1_0(M)), ~(\Box+\Lambda)\bw = 0\}.
\end{equation}
\end{thm}
\begin{proof} Given $\bgamma \in \mathscr{T}(M)$ let $\bgamma^{\prime}=\bgamma+\pounds_{\bw}\bg$ for an arbitrary $\bw\in C^\infty_{SC}(T^1_0(M))$. Taking the divergence of the trace-reversal of this and using Lemma~\ref{lem:divtracereversegauge} gives
\begin{equation}
\nabla^{a}\overline{\gamma}^{\prime}_{ab}=\nabla^{a}\overline{\gamma}_{ab}+(\Box+\Lambda)w_{b}.
\end{equation}
Therefore $\bgamma^{\prime}\in \mathscr{T}^{dD}(M)$  if and only if $\bw$ obeys
\begin{equation}
\label{eqn:ddgaugehypeqn}
(\Box+\Lambda) w^{b} =-\nabla^{a}\overline{\gamma}_{a}^{\phantom{a}b}.
\end{equation} 
Applying standard results for hyperbolic equations, such as \cite[Thm~3.2.11]{Bar} (generalized to permit non-compactly supported source terms as in Corollary~5 in \cite[Ch.~3]{QFTCSTBar}), there is a unique solution to \eqref{eqn:ddgaugehypeqn} with initial data $\bw|_{\Sigma}$ and $\nabla_{\bn}\bw|_{\Sigma}$ that are smooth and compactly supported (we take them identically vanishing) on a smooth spacelike Cauchy surface $\Sigma$ with future-pointing unit normal vector $\bn$. The solution is smooth and spacelike-compact.\footnote{Here we use the fact that $\bgamma$ has
spacelike-compact support, which is not obviously the same as having support that has compact intersection with each Cauchy surface. After this 
paper was completed, and prompted by this concern, Sanders has shown that they are  equivalent~\cite{Sanders_compact}.} Hence $\bgamma'-\bgamma\in \mathscr{G}(M)$, which gives the splitting \eqref{eqn:dDsplit}. 
The final statement follows easily.\hspace{0.5cm}
\end{proof}
{\noindent\em Remark.}  From the proof, it is clear that the vector field $\bw$ may be chosen to have
support contained in the future (or past) of any given spacelike Cauchy surface.
The space $\mathscr{G}^{dD}(M)$ specifies the residual gauge freedom within the de Donder class.

Linearity of the equation of motion combined with the previous theorem gives:
\begin{cor}
\label{cor:deDsplitS}
The space $\mathscr{S}(M)$ decomposes as
\begin{equation}
\mathscr{S}(M)=\mathscr{S}^{dD}(M)+\mathscr{G}(M),
\end{equation}
where $\mathscr{S}^{dD}(M)=\{\bgamma \in \mathscr{T}^{dD}(M) \hspace{0.1cm} | \hspace{0.1cm} L_{ab}(\bgamma)=0\}$
is the space of de Donder gauge solutions. Moreover, 
$\mathscr{S}^{dD}(M)\cap \mathscr{G}(M) =\mathscr{G}^{dD}(M)$.
\end{cor}

Next, we define a partial differential operator
\begin{equation}
\label{eqn:Pdef}
P_{ab}^{\phantom{ab}cd} := 
\Box \delta^{c}_{\phantom{c}a}\delta^{d}_{\phantom{d}b} -2R^{c \phantom{ab}d}_{\phantom{c}ab},
\end{equation}
which is of the type considered by Lichnerowicz in \cite{Lichnerowicz} (see equation~{(10.4)} of that reference). We now note some important identities concerning this operator.
\begin{lem}
\label{lem:traceoperatorcommute}
On a cosmological vacuum background spacetime, $P$ commutes with trace reversal. 
In particular, $P(\overline{\bgamma})=0$ if and only if $P(\bgamma)=0$.
\end{lem}
\begin{proof} We compute
\begin{align}
\label{eqn:pcabdoverlineexpand}
\overline{(\Box \delta^{c}_{\hspace{0.1cm}a}\delta^{d}_{\hspace{0.1cm}b} -2R^{c \hspace{0.2cm}d}_{\hspace{0.1cm}ab})f_{cd}}&=(\Box \delta^{c}_{\hspace{0.1cm}a}\delta^{d}_{\hspace{0.1cm}b} -2R^{c \hspace{0.2cm}d}_{\hspace{0.1cm}ab})f_{cd}-\frac{1}{2}g_{ab}(\Box f +2\Lambda f) \notag\\
&=\Box(f_{ab}-\frac{1}{2}g_{ab}f)-2R^{c \hspace{0.2cm}d}_{\hspace{0.1cm}ab}(f_{cd}-\frac{1}{2}g_{cd}f)
\end{align}
by using the identities 
$g_{ab}g^{ef}R^{c \hspace{0.2cm}d}_{\hspace{0.1cm}ef}f_{cd}=-g_{ab}R^{cd}f_{cd}=-\Lambda g_{ab}f$ and $g_{cd}R^{c \hspace{0.2cm}d}_{\hspace{0.1cm}ab}f = -R_{ab}f=-\Lambda g_{ab} f = g_{ab}g^{ef}R^{c \hspace{0.2cm}d}_{\hspace{0.1cm}ef}f_{cd}$, which hold in cosmological vacuum spacetimes.\,\,
\end{proof}

\begin{thm}
\label{thm:Lident}
For any $\bgamma\in C^\infty(S^0_2(M))$, on a cosmological vacuum background spacetime, 
\begin{equation}\label{eq:Lident}
2L_{ab}(\bgamma) =- P_{ab}^{\phantom{ab}cd} \overline{\gamma}_{cd} +(\overline{\pounds_{(\nabla\cdot \overline{\bgamma})^{\sharp}}\bg})_{ab}
\end{equation}
or equivalently
\begin{equation}
2\overline{L_{ab}(\bgamma)} =-P_{ab}^{\phantom{ab}cd}\gamma_{ab}+(\pounds_{(\nabla\cdot \overline{\bgamma})^{\sharp}}\bg)_{ab}.
\end{equation}
\end{thm}
\begin{proof} The Lie derivative term is
\begin{equation}
(\overline{\pounds_{(\nabla\cdot \overline{\bgamma})^{\sharp}}\bg})_{ab} = \nabla_{a}\nabla^{c}\gamma_{cb}+\nabla_{b}\nabla^{c}\gamma_{ca}-\nabla_{a}\nabla_{b}\gamma-g_{ab}(\nabla^{d}\nabla^{c}\gamma_{cd}-\frac{1}{2}\Box \gamma).
\end{equation}
One may show, using \eqref{eqn:Riccitenscalvalue}, that $\nabla_{a}\nabla^{c}\gamma_{bc}=\nabla^{c}\nabla_{a}\gamma_{bc}-\Lambda \gamma_{ab}-R^{c \phantom{ab}d}_{\phantom{c}ab}\gamma_{cd}$ and so
\begin{align}
\label{eqn:liederivnabgam}
(\overline{\pounds_{(\nabla\cdot \overline{\bgamma})^{\sharp}}\bg})_{ab} &= \nabla^{c}\nabla_{a}\gamma_{bc}+\nabla^{c}\nabla_{b}\gamma_{ac}-2\Lambda \gamma_{ab}-2R^{c \phantom{ab}d}_{\phantom{c}ab}\gamma_{cd}
-\nabla_{a}\nabla_{b}\gamma \notag\\ 
&\qquad\qquad
-g_{ab}(\nabla^{d}\nabla^{c}\gamma_{cd}-\frac{1}{2}\Box \gamma).
\end{align}
The $P(\overline{\bgamma})$ term is 
\begin{equation}
\label{eqn:Povergammaexplic}
- P_{ab}^{\phantom{ab}cd} \overline{\gamma}_{cd} = -\Box \gamma_{ab}+\frac{1}{2}g_{ab}\Box \gamma+2R^{c \phantom{ab}d}_{\phantom{c}ab} \gamma_{cd}+\Lambda g_{ab} \gamma.
\end{equation}
Combine \eqref{eqn:liederivnabgam} and \eqref{eqn:Povergammaexplic} and then compare with \eqref{eqn:symmL}. The second identity follows from the first by using~Lemma~\ref{lem:traceoperatorcommute}. \end{proof}

Therefore, for linearized gravity solutions on cosmological vacuum spacetimes that satisfy the de Donder condition $\nabla \cdot \overline{\bgamma}=0$, the equation of motion \eqref{eqn:nonhypeqn} reduces to
\begin{equation}
\label{eqn:overlinelineqngauge}
P_{ab}^{\phantom{ab}cd} \overline{\gamma}_{cd} = \Box \overline{\gamma}_{ab} - 2 R^{c\phantom{ab}d}_{\phantom{c} ab} \overline{\gamma}_{cd} =0,
\end{equation}
or equivalently, by Lemma~\ref{lem:traceoperatorcommute},
\begin{equation}
\label{eqn:ddeqnmotion}
\Box \gamma_{ab} -2R^{c \hspace{0.2cm}d}_{\hspace{0.1cm}ab}\gamma_{cd}=0.
\end{equation}

 For future reference, we note the following identity.
\begin{lem}
\label{lem:nablaP}
For any $\bgamma\in C^\infty(T^0_2(M))$, on a cosmological vacuum background spacetime,
\begin{equation}\label{eq:divPident}
\nabla^a (P_{ab}^{\phantom{ab}cd}\gamma_{cd}) = (\Box + \Lambda)\nabla^a \gamma_{ab}.
\end{equation}
\end{lem}
\begin{proof} Expanding out the left-hand side of \eqref{eq:divPident} gives
\begin{equation}
\nabla^a (P_{ab}^{\phantom{ab}cd}\gamma_{cd}) = \nabla^{a}\Box \gamma_{ab}-2(\nabla^{a}R^{c \phantom{ab}d}_{\phantom{c}ab}) \gamma_{cd} -2R^{c \phantom{ab}d}_{\phantom{c}ab} \nabla^{a}\gamma_{cd}.
\end{equation}
Substituting \eqref{eqn:Riccitenscalvalue} into the contracted Bianchi identity $
\nabla_{a}R_{bcd}^{\hspace{0.4cm}a}+\nabla_{b}R_{cd}-\nabla_{c}R_{bd}=0$
gives  $\nabla_{a}R_{bcd}^{\phantom{bcd}a}=0$ and hence
\begin{equation}
\nabla^a (P_{ab}^{\phantom{ab}cd}\gamma_{cd}) = \nabla^{a}\Box \gamma_{ab} -2R^{c \phantom{ab}d}_{\phantom{c}ab} \nabla^{a}\gamma_{cd}.
\end{equation}
By using the Riemann tensor identity~\eqref{eqn:Riemann} and the restriction to cosmological vacuum spacetimes~\eqref{eqn:Riccitenscalvalue} one can show that
\begin{equation}
\nabla^{a}\Box \gamma_{ab}= \Box \nabla^{a}\gamma_{ab}+\Lambda \nabla^{a}\gamma_{ab}+2R^{ad \phantom{b}c}_{\phantom{a}\phantom{d}b}\nabla_{d}\gamma_{ac}
\end{equation}
and hence the result. 
\end{proof}

It follows that if $\bgamma$ is {\em any}  solution to \eqref{eqn:ddeqnmotion}, its divergence, and the divergence of its trace reverse, obey 
\begin{equation}\label{eqn:diveqnmotion}
(\Box+\Lambda)(\nabla^{a} \gamma_{ab}) =
(\Box+\Lambda)(\nabla^{a} \overline{\gamma}_{ab})=0.
\end{equation}
In addition, we see directly from \eqref{eqn:ddeqnmotion} that its trace obeys
\begin{equation}
\label{eqn:traceeqnmotion}
(\Box+2\Lambda)\gamma=0.
\end{equation}
In deriving equations \eqref{eqn:overlinelineqngauge},~\eqref{eqn:ddeqnmotion}, \eqref{eqn:diveqnmotion} and  \eqref{eqn:traceeqnmotion}, the result \eqref{eqn:Riccitenscalvalue} was used.

\subsection{Transverse-traceless gauge}

Many discussions of linearized gravity employ the transverse-traceless (TT) gauge, $\nabla^{a}\gamma_{ab} =0$ and $\gamma =0$, which is de Donder gauge with the additional constraint of vanishing trace. Inspecting the proof
of Theorem~\ref{thm:deDsplit}, we see that $\bgamma\in\mathscr{T}(M)$ can be put into the 
TT gauge if and only if the system
\begin{equation}
\label{eqn:remDDw_and_traceconstraintw}
(\Box+\Lambda)w^{a}=0,\qquad 
\nabla_{a}w^{a}=-\frac{1}{2}\gamma
\end{equation}
can be solved for $\bw\in C^\infty_{SC}(T^1_0(M))$.  It turns out that for vacuum spacetimes with a non-vanishing cosmological constant (e.g. de Sitter) this may be achieved at least when $\bgamma$ 
solves the linearized Einstein equation. 
\begin{thm}
\label{thm:TTsplit}
For cosmological vacuum spacetimes with $\Lambda \neq 0$ one may perform the following decomposition of the space of spacelike-compact solutions:
\begin{equation}
\label{eqn:TTsplit}
\mathscr{S}(M) = \mathscr{S}^{TT}(M) + \mathscr{G}(M).
\end{equation}
\end{thm}
As this is a departure from our main theme, the details are relegated to Appendix~\ref{appx:TT}. 

By contrast, for vacuum spacetimes with a vanishing cosmological constant we find that there is a cohomological obstruction to the solution of \eqref{eqn:remDDw_and_traceconstraintw}: it is possible
if and only if the trace of the solution $\bgamma$ obeys
\begin{equation}
\int_\Sigma \nabla_{\bn}\gamma\,d\Sigma =0
\end{equation}
on some (and hence all) Cauchy surfaces. To some extent this problem arises because we only consider spacelike-compact perturbations. If this restriction is dropped {\em and} $\Sigma$ is non-compact, then the TT gauge can also be achieved for $\Lambda=0$, as is the case in Minkowski space.

\subsection{Synchronous gauge}
\label{sec:synchronous}

The synchronous gauge is defined relative to a spacelike Cauchy surface $\Sigma$ by the condition $n^{a}\gamma_{ab}=0$, where $\bn$ is the future-pointing unit normal vector to $\Sigma$. Here, we describe how one can make a gauge transformation to put any solution into the synchronous gauge near a Cauchy surface. The result is similar to \cite[Lem.~1.1]{AshtekarMagnon}, which shows that this
can be done {\em on} a Cauchy surface; our result is therefore slightly more general and our proof treats the solution of various equations arising in  detail.

Before we state the theorem we recall some definitions. Given a submanifold $S \subset M$, the normal exponential map $\exp^{\perp}$ is the restriction of the exponential map to the normal bundle of $S$; therefore all the geodesics arising from this map will meet $S$ orthogonally. A normal neighbourhood of $S$ is a neighbourhood of $S$ that is diffeomorphic under $\exp^{\perp}$ to a connected neighbourhood of the zero section in the normal bundle of $S$. In particular, on any normal neighbourhood $\mathscr{O}$ of $\Sigma$ there is a unique future-pointing, geodesic, hypersurface orthogonal, unit vector field that we call the normal field of $\Sigma$ in $\mathscr{O}$.
\begin{thm}
\label{thm:synchronous}
Let $\Sigma$ be a smooth spacelike Cauchy surface with future-pointing unit normal vector $\bn$. Let $\mathscr{O}$ be any open normal neighbourhood of $\Sigma$, whose closure is contained in another normal neighbourhood of $\Sigma$. Then 
\begin{equation}
\mathscr{T}(M) = \mathscr{T}^{synch}_{\Sigma, \mathscr{O}}(M) + \mathscr{G}(M),
\end{equation}
where $\mathscr{T}^{synch}_{\Sigma, \mathscr{O}}(M) = \{ \bgamma \in \mathscr{T}(M) \hspace{0.2cm} | \hspace{0.2cm} \tilde{n}^{a}\gamma_{ab}=0 \hspace{0.2cm} \text{on} \hspace{0.2cm} \mathscr{O} \}$ and $\tilde{\bn}$ is the normal field of $\Sigma$ in $\mathscr{O}$.
In particular, $n^a \gamma_{ab}=0$ holds on $\Sigma$. 
\end{thm}
{\noindent\em Remarks.} $\Sigma$ has normal neighbourhoods by \cite[Prop.~7.26]{ONeill}. Given any such normal neighbourhood, we may restrict to a smaller normal neighbourhood whose closure is contained in the original. Therefore the existence of $\mathscr{O}$ in the hypothesis is not restrictive.

\begin{proofof}{Theorem~\ref{thm:synchronous}}
Let $\bgamma \in \mathscr{T}(M)$ be arbitrary. Then the condition $\bgamma+\pounds_{\bw}\bg \in \mathscr{T}^{synch}_{\Sigma, \mathscr{O}}(M)$ amounts to the equations
\begin{align}
\label{eqn:W0eqn}
\nabla_{\tilde{\bn}}W_{0}& =-\frac{1}{2}\tilde{n}^{a}\tilde{n}^{b}\gamma_{ab}
\\
\label{eqn:wpara}
(\nabla_{\tilde{\bn}}w_{\parallel})_{b}-w_{\parallel a}\nabla_{b}\tilde{n}^{a} &= -\tilde{n}^{a}\gamma_{ab}-\nabla_{b}W_{0}-\frac{1}{2}\tilde{n}^{a}\tilde{n}^{c}\gamma_{ac}\tilde{n}_{b},
\end{align}
where $W_{0}=\tilde{n}^{a}w_{a}$ and $\bw_{\parallel}=\bw+W_{0}\tilde{\bn}^{\flat}$.

The first step is to obtain a solution $W_{0}$ to \eqref{eqn:W0eqn} on $\mathscr{O}$. This can be achieved as follows. Through each point $p \in \Sigma$ we have a unit speed normal geodesic $\lambda_{p}:I\to\mathscr{O}$ with $0 \in I \subset \mathbb{R}$ and $\lambda_{p}(0)=p$. Equation \eqref{eqn:W0eqn} can be integrated along $\lambda_{p}$ to give a solution
\begin{equation}
\label{eqn:W0soln}
(W_{0}\circ\lambda_{p})(t)=-\frac{1}{2}\int^{t}_{0}{(\tilde{n}^{a}\tilde{n}^{b}\gamma_{ab} \circ \lambda_{p})(s)ds}
\end{equation}
and we define the scalar function $W_{0}$ at any $q\in\mathscr{O}$ by 
$W_{0}(q):=(W_{0}\circ\lambda_{p_{q}})(t_{q})$, where $p_q$ and $t_q$
are uniquely determined by $q=\lambda_{p_q}(t_q)$.  This satisfies \eqref{eqn:W0eqn} by definition. It is smooth on $\mathscr{O}$ because $t_{q}$ and $p_{q}$ vary smoothly with $q$ under the normal exponential map, which is a diffeomorphism on $\mathscr{O}$, and $\tilde{n}^{a}\tilde{n}^{b}\gamma_{ab}$ is smooth by assumption.

Now, to obtain $\bw_{\parallel}$ on $\mathscr{O}$ we solve \eqref{eqn:wpara} locally within a neighbourhood of each geodesic and then patch together the results with a partition of unity. 
For each $q\in \Sigma$, let $\mathscr{N}_{q} \subset \Sigma$ be an open normal neighbourhood of $q$. Hence on $\mathscr{N}_{q}$ we have well-defined normal coordinates $x^i$ ($i=1,2,3$) based at $q$ and associated  basis vector fields $\be_{i}$. Now, for each $q\in\Sigma$ let $\mathscr{M}_{q}$ be the open set
of points in $\mathscr{O}$ connected to $\Sigma$ by geodesics emanating normally from $\mathscr{N}_q$. The sets $\mathscr{M}_{q}$, for $q\in\Sigma$, form an open cover for $\mathscr{O}$ because it is a normal neighbourhood.  Within each $\mathscr{M}_{q}$ we can introduce Gaussian normal coordinates given by: the proper time $t$  along the geodesics, with $t=0$ on $\Sigma$, and the
normal coordinates $x^i$ mentioned above. In these coordinates, \eqref{eqn:wpara} becomes
\begin{equation}
\label{eqn:dwidt}
\frac{d(w_{\parallel})_{i}(t,x)}{dt}-2\Gamma^{j}_{\hspace{0.1cm}i0}(t,x)(w_{\parallel})_{j}(t,x)=-\gamma_{0i} - \frac{\partial W_{0}(t,x)}{\partial x^{i}}.
\end{equation}
This system can be solved using standard results (see, e.g., \cite[Sec.~1.6]{Taylor}) to give $\bw_{\parallel}$ on $\mathscr{M}_{q}$, where we have also used that $(w_{\parallel})_{0}=0$ in these coordinates. This process is repeated on each $\mathscr{M}_{q}$ for all $q \in \Sigma$. 

As $\Sigma$ is an embedded submanifold of $M$, it will also be second-countable and Hausdorff. Therefore by \cite[Thm~1.11]{Warner}, the open cover $\{ \mathscr{N}_{q} \hspace{0.1cm}|\hspace{0.1cm} q \in \Sigma\}$ of $\Sigma$ by normal neighbourhoods will admit a countable partition of unity $\{\chi_{\lambda} \hspace{0.1cm}|\hspace{0.1cm}\lambda\in I\}$ subordinate to the cover with $\supp\chi_{\lambda}$ compact for each $\lambda \in I$. Hence, for each $\lambda\in I$ there exists a $q\in \Sigma$ such that $\supp\chi_{\lambda}\subset \mathscr{N}_{q}$. To obtain a suitable partition of unity $\tilde{\chi}_{\lambda}$ on the $\mathscr{M}_{q}$'s, we solve $\nabla_{\tilde{\bn}}\tilde{\chi}_{\lambda}=0$ with $\tilde{\chi}_{\lambda}|_{\Sigma}=\chi_{\lambda}$ 
by integrating along integral curves of $\tilde{\bn}$ as before.

Therefore the $\bw_{\parallel}$ that we seek on $\mathscr{O}$ is given by
\begin{equation}
\label{eqn:chiwpara}
\bw_{\parallel} = \sum_{\lambda}{\tilde{\chi}_{\lambda}\bw_{\parallel}^{\lambda}},
\end{equation}
where each $\bw_{\parallel}^{\lambda}$ is the solution to \eqref{eqn:dwidt} on the set $\mathscr{M}_{q}$ that contains $\tilde{\chi}_{\lambda}$. Observe that \eqref{eqn:chiwpara} satisfies \eqref{eqn:wpara} on $\mathscr{O}$ by the properties of $\tilde{\chi}_{\lambda}$.

In conjunction with $W_{0}$ this will give the $\bw$ on $\mathscr{O}$ to transform to the synchronous gauge. We now examine the support properties of $\bw$. Outside $\supp\bgamma$, \eqref{eqn:W0eqn} reduces to $\nabla_{\bn}W_{0}=0$ and so $W_{0}=constant$ along each normal geodesic emanating from $\Sigma$, as long as the geodesic does not enter $\supp\bgamma$. Choosing $W_{0}|_{\Sigma}=0$ yields $W_{0}=0$ on every geodesic that does not intersect $\supp \bgamma$; hence $W_{0}|_{\mathscr{O}}$ is spacelike-compact. Using this means that outside $\supp \bgamma$ equation \eqref{eqn:wpara} reduces to $(\nabla_{\tilde{\bn}}w_{\parallel})_{b}-w_{\parallel a}\nabla_{b}\tilde{n}^{a}=0$; in Gaussian normal coordinates, the right-hand side of \eqref{eqn:dwidt} vanishes. Thus with $\bw_{\parallel}|_{\Sigma}=0$, the solution $w_{\parallel i}$ vanishes
in every $\mathscr{M}_q$ that does not intersect $\supp \bgamma$, so $\bw|_{\mathscr{O}}$ is spacelike-compact. Outside  $\mathscr{O}$ we let $\bw$ smoothly decay to zero. Hence $\bw$ may
be chosen to be compactly supported and therefore $\pounds_{\bw}\bg \in \mathscr{G}(M)$.  
\end{proofof}

\section{Existence and Uniqueness of solutions}
\label{sec:ExistenceUniqueness}

We will now prove existence and uniqueness up to gauge of solutions to the non-hyperbolic equation \eqref{eqn:nonhypeqn}, and show how Green's operators of the de Donder hyperbolic equation \eqref{eqn:ddeqnmotion} are related to linearized gravity solutions.

\subsection{Existence}

Let $\Sigma$ be a smooth spacelike Cauchy surface with unit future-pointing normal vector $\bn$. The Cauchy data map\footnote{The notation $T^{0}_{2}(M)|_{\Sigma}$ means the restriction of rank $(0,2)$ tensor fields on $M$ to the surface $\Sigma$.} $\Data_\Sigma:C^\infty(T^0_2(M))\to C^\infty((T^0_2(M)|_\Sigma)\oplus
C^\infty((T^0_2(M)|_\Sigma)$ is defined by
\begin{equation}
\Data_\Sigma(\bgamma) := (\bgamma|_\Sigma,\nabla_\bn \bgamma|_\Sigma).
\end{equation}
Our aim is to solve the equation $L_{ab}(\bgamma)=0$ subject to given 
$\Data_\Sigma(\bgamma)$. As is well known, this cannot be achieved
for arbitrary data, because the components $L_{ab}(\bgamma)n^{b}|_{\Sigma}$
do not involve second time-derivatives of $\bgamma$ and are therefore completely determined by the initial data. Put more formally, there is a 
constraint map $\bC^{\Sigma}:C^\infty((T^0_2(M)|_\Sigma)\oplus
C^\infty((T^0_2(M)|_\Sigma) \to C^{\infty}(T^{0}_{1}(M)|_{\Sigma})$, 
so that
\begin{equation}
\label{eqn:IVCbC}
\bC^{\Sigma}(\bgamma|_\Sigma,\nabla_\bn \bgamma|_\Sigma):=n^{a}L_{ab}(\bgamma)|_{\Sigma}
\end{equation}
and the Cauchy data must be restricted to the kernel of $\bC^\Sigma$. 
The precise form of the linear map $\bC^\Sigma$ will not be needed here. 
We observe that gauge invariance of $L_{ab}$ entails that $\bC^{\Sigma} \circ \Data_\Sigma$ is also gauge invariant.

\begin{thm}
\label{thm:existence}
Let $\Sigma$ be a smooth spacelike Cauchy surface with future-pointing unit normal vector $\bn$. For any initial data $\bzeta$, $\bxi \in C_0^{\infty}(S^0_2(M)|_{\Sigma})$ satisfying the initial value constraint $\bC^{\Sigma}(\bzeta,\bxi)=0$ there exists a solution $\bgamma\in\mathscr{T}(M)$ to \eqref{eqn:nonhypeqn}
such that $\Data_\Sigma(\bgamma)=(\bzeta,\bxi)$.
\end{thm}
\begin{proof} The proof is broken into two steps. First, we make a gauge transformation that puts the initial data into the de Donder gauge on $\Sigma$. 
Second, as the initial value problem for \eqref{eqn:ddeqnmotion} is
well-posed, we obtain a solution with the transformed data, which obeys the de Donder condition globally; it is here that the initial value constraint is vital. By Theorem~\ref{thm:Lident}, the solution will satisfy the linearized Einstein equation \eqref{eqn:nonhypeqn} and so by undoing the original gauge transformation we obtain a solution with the original Cauchy data.

Following the structure just set out, we begin by constructing a smooth 
$\bchi\in\mathscr{T}(M)$ such that
$\Data_\Sigma(\bchi)=(\bzeta,\bxi)$. This is accomplished as follows: In a normal neighbourhood of $\Sigma$, we use parallel transport along geodesics normal to $\Sigma$ to obtain $\tilde{\bzeta}\in \mathscr{T}(M)$ 
with $\Data_\Sigma(\tilde{\bzeta})=(\bzeta,0)$. 
Taking any extension of $\bxi$ in $\mathscr{T}(M)$, we form $\bchi=
\tilde{\bzeta}+s\bxi$, where $s$ is determined uniquely at each point $p$ of the 
normal neighbourhood by the requirement $p=\exp_q s n|_q$ for some $q\in\Sigma$.
By extending smoothly, we obtain $\bchi\in \mathscr{T}(M)$ with the required properties. 

Using the splitting of Theorem~\ref{thm:deDsplit} there exists $\bw\in C^\infty_{SC}(T^1_0(M))$ such that $\tilde{\bgamma} =\bchi +\pounds_{\bw}\bg$ 
obeys the de Donder condition $\nabla^{a} \overline{\tilde{\gamma}}_{ab}=0$ and
\begin{align}
(\tilde{\gamma}_{ab}-2\nabla_{(a}w_{b)})|_{\Sigma} =\chi_{ab}|_\Sigma &=\zeta_{ab} \notag\\
n^{c}\nabla_{c}(\tilde{\gamma}_{ab}-2\nabla_{(a}w_{b)})|_{\Sigma}
=n^{c}\nabla_{c}\chi_{ab}|_\Sigma  &= \xi_{ab}
.\label{eqn:modifiedICs}
\end{align}

Passing to the second step, let $\hat{\bgamma}\in \mathscr{T}(M)$ be the (unique) solution to the hyperbolic equation $\Box \hat{\gamma}_{ab} - 2 R^{c\phantom{ab}d}_{\phantom{c}ab} \hat{\gamma}_{cd}=0$ with initial data $\Data_\Sigma (\hat{\bgamma}) =\Data_\Sigma(\tilde{\bgamma})$.  The existence
and uniqueness of $\hat{\bgamma}$ follows from \cite[Thm~3.2.11]{Bar}.

The key point is now to show that $\hat{\bgamma}$ obeys the de Donder condition globally. As $\hat{\bgamma}$ obeys \eqref{eqn:ddeqnmotion} then by Lemma~\ref{lem:nablaP}, $\nabla^{a} \overline{\hat{\gamma}}_{ab}$ obeys the hyperbolic equation  \eqref{eqn:diveqnmotion}; it also vanishes on $\Sigma$ because
$\nabla^{a}\overline{\hat{\gamma}}_{ab}|_{\Sigma} = \nabla^{a}\overline{\tilde{\gamma}}_{ab}|_{\Sigma}=0$.
The following identity is now required.
\begin{lem} On a cosmological vacuum background spacetime, for any solution $\hat{\bgamma}$ to \eqref{eqn:ddeqnmotion},
it holds that
$n^{c}\nabla_{c}( \nabla^{a}\overline{\hat{\gamma}}_{ab})|_\Sigma=2L_{ab}(\hat{\bgamma})n^{a}|_\Sigma$.
\end{lem}
\begin{proof} Combining the hypothesis with Lemma~\ref{lem:traceoperatorcommute} and using Theorem~\ref{thm:Lident} gives $2L_{ab}(\hat{\bgamma})=(\overline{\pounds_{(\nabla\cdot \overline{\hat{\bgamma}})^{\sharp}}\bg})_{ab}$. Contracting with $\bn$ and expanding the right-hand side gives
\begin{equation}
2L_{ab}(\hat{\bgamma})n^{a} = n^{a}\nabla_{a}\nabla^{c}\overline{\hat{\gamma}}_{cb}+n^{a}\nabla_{b}\nabla^{c}\overline{\hat{\gamma}}_{ca}-n_{b}\nabla^{d}\nabla^{c}\overline{\hat{\gamma}}_{cd}.
\end{equation}
The metric may be written in terms of the normal vector $\bn$ and a projection operator (see \cite[Ch.~2.7]{HawkingEllis}), so that $g_{ab}=-n_{a}n_{b}+q_{ab}$. This allows one to split a vector into its components normal and tangential to $\Sigma$. Using that $\nabla^{a}\overline{\hat{\gamma}}_{ab}|_{\Sigma}=0$ and $n_{a}n^{a}=-1$ one finds that
\begin{equation}
2L_{ab}(\hat{\bgamma})n^{a}|_{\Sigma}=n^{c}\nabla_{c}( \nabla^{a}\overline{\hat{\gamma}}_{ab})|_\Sigma.
\end{equation}
\end{proof}

As stated immediately prior to Theorem~\ref{thm:existence}, the constraints are gauge invariant and so we have the following chain of equalities
\begin{equation}
L_{ab}(\hat{\bgamma})n^{a}|_\Sigma=\bC^{\Sigma}(\Data_{\Sigma}(\hat{\bgamma}))=\bC^{\Sigma}(\Data_{\Sigma}(\tilde{\bgamma}))=\bC^{\Sigma}(\Data_{\Sigma}(\bchi))=0
\end{equation}
and thus $n^{c}\nabla_{c}( \nabla^{a}\overline{\hat{\gamma}}_{ab})|_\Sigma=0$.

Accordingly, we have shown that $\nabla^{a} \overline{\hat{\gamma}}_{ab}$ 
obeys a hyperbolic equation with vanishing initial data on $\Sigma$; it therefore
vanishes globally in $M$ by \cite[Cor.~3.2.4]{Bar}. Thus, $\hat{\bgamma}$ solves \eqref{eqn:ddeqnmotion} and satisfies the de Donder condition, so $\hat{\bgamma}\in\mathscr{S}^{dD}(M)$. By undoing the original gauge transformation, we obtain a solution $\bgamma=\hat{\bgamma} - \pounds_\bw\bg$ to the linearized Einstein equation \eqref{eqn:symmL}, which obeys $\Data_\Sigma(\bgamma)=(\bzeta,\bxi)$ by virtue of \eqref{eqn:modifiedICs}. \end{proof}

\subsection{Uniqueness}

Given any initial data satisfying the constraints then by Theorem~\ref{thm:existence} there exists a solution to the linearized Einstein equation. However, as the next theorem shows, this solution is only unique up to gauge equivalence.
\begin{thm}
\label{thm:uniqueness} Suppose $\bgamma,\bgamma^{\prime} \in \mathscr{S}(M)$ with $\Data_\Sigma (\bgamma)=\Data_\Sigma (\bgamma^{\prime})$ on some spacelike Cauchy surface $\Sigma$. Then $\bgamma=\bgamma^{\prime}+\pounds_{\bw}\bg$, also written as $\bgamma \thicksim \bgamma^{\prime}$. If, additionally, $\bgamma,\bgamma^{\prime} \in \mathscr{S}^{dD}(M)$ then gauge equivalence is replaced by equality.
\end{thm}
\begin{proof} Let $\bxi= \bgamma-\bgamma^{\prime}$ 
which satisfies $\Data_\Sigma (\bxi)=0$. By Theorem~\ref{thm:deDsplit} we may write $\tilde{\bgamma}=\bxi+\pounds_\bw\bg$ 
where $\tilde{\bgamma}\in\mathscr{S}^{dD}(M)$ and $\bw\in C^\infty_{SC}(T^1_0(M))$ obeys $(\Box+\Lambda)w^{b} =-\nabla^{a} \overline{\xi}_{a}^{\phantom{a}b}$. Choose initial data,  $\bw |_{\Sigma}=0$ and $\nabla_{\bn}\bw|_{\Sigma}=0$.
Therefore $\Data_{\Sigma}(\tilde{\bgamma})$ is given by
\begin{align}
\tilde{\gamma}_{ab}|_{\Sigma} &=(\nabla_{a}w_{b}+\nabla_{b}w_{a})|_{\Sigma} =0 \\
n^{c}\nabla_{c}\tilde{\gamma}_{ab}|_{\Sigma} &=n^{c}\nabla_{c}(\nabla_{a}w_{b}+\nabla_{b}w_{a})|_{\Sigma}. \label{eq:data2}
\end{align}
In fact $n^{c}\nabla_{c}\tilde{\gamma}_{ab}|_{\Sigma}$ also vanishes, as we will now show. Firstly, by the choice of data we know that $\nabla_{a}w_{b}|_{\Sigma}=0$ and therefore any derivative of this taken tangentially to $\Sigma$ will vanish. Next, 
using the Riemann tensor identity \eqref{eqn:Riemann} and the conditions on $\bw$ at $\Sigma$, we have
\begin{align}
n^{c}\nabla_{c}\nabla_{a}w_{b}|_{\Sigma}&=\nabla_{a}(n^{c}\nabla_{c}w_{b})|_{\Sigma}-(\nabla_{a}n^{c})\nabla_{c}w_{b}|_{\Sigma}+n^c R_{cab}^{\phantom{cab}d}w_{d}|_{\Sigma} \\
&= \nabla_{a}(n^{c}\nabla_{c}w_{b})|_{\Sigma} =-n_{a}n^{d}\nabla_{d}(n^{c}\nabla_{c}w_{b})|_{\Sigma}.
\end{align}
As $\Data_{\Sigma}(\bxi)=0$, we have
$(\Box+\Lambda)\bw|_{\Sigma} = -\nabla\cdot\overline{\bxi}|_{\Sigma}
= 0$. Expanding, using
$g_{ab}=-n_{a}n_{b}+q_{ab}$,
\begin{equation}
n^{a}\nabla_{a}(n^{c}\nabla_{c}w^{b})|_{\Sigma} =  q^{ac}\nabla_{a}\nabla_{c}w^{b}|_{\Sigma}+(n^{a}\nabla_{a}n^{c})\nabla_{c}w^{b}|_{\Sigma}+\Lambda w^{b}|_{\Sigma}
=0.
\end{equation}
Therefore $n^{c}\nabla_{c}\nabla_{a}w_{b}|_{\Sigma}=0$ and hence by \eqref{eq:data2}, $n^{c}\nabla_{c}\tilde{\gamma}_{ab}|_{\Sigma}$ vanishes.
As $\tilde{\bgamma}$ satisfies the hyperbolic equation \eqref{eqn:ddeqnmotion} with vanishing Cauchy data, it vanishes globally by \cite[Cor.~3.2.4]{Bar}. Thus $\bxi \thicksim 0$ and hence $\bgamma \thicksim \bgamma^{\prime}$.

If both of the solutions $\bgamma$ and $\bgamma^{\prime}$ are de Donder with the same initial data then they must coincide because they solve the hyperbolic equation \eqref{eqn:ddeqnmotion} with identical data. \end{proof}

\subsection{Green's operators}
\label{sec:greensoperators}

Any de Donder solution satisfies the two equivalent hyperbolic wave equations \eqref{eqn:overlinelineqngauge} and \eqref{eqn:ddeqnmotion}. By the results of \cite[Ch.~1.5]{Bar}, the differential operator $P$ (defined in equation~\eqref{eqn:Pdef}) is a normally hyperbolic operator and therefore, by \cite[Cor.~3.4.3]{Bar}, admits unique advanced~($-$) and retarded~($+$) Green's operators $\bE^{\pm}: C^{\infty}_{0}(T^{0}_{2}(M)) \to C^{\infty}(T^{0}_{2}(M))$, whose action on test tensors we can write as
\begin{equation}
(E^{\pm \hspace{0.2cm}c^{\prime}d^{\prime}}_{\hspace{0.1cm}ab}f_{c^{\prime}d^{\prime}})(x)=\int_{M}{E^{\pm \hspace{0.2cm}c^{\prime}d^{\prime}}_{\hspace{0.1cm}ab}(x,x^{\prime})f_{c^{\prime}d^{\prime}}(x^{\prime})dvol(x^{\prime})},
\end{equation}
where $\bfab \in C^{\infty}_{0}(T^{0}_{2}(M))$. The operators $\bE^{\pm}$ satisfy
$P (\bE^\pm\bfab) = \bE^\pm P (\bfab)= \bfab$ for all $\bfab \in C^{\infty}_{0}(T^{0}_{2}(M))$, and have the support properties $\supp(\bE^{\pm} \bfab) \subset J^{\pm}(\supp \bfab)$;
moreover $\bgamma^\pm = \bE^\pm \bfab$ is the unique solution to $P(\bgamma^\pm) =\bfab$
with support that is compact to the past ($+$)/future ($-$) (for the uniqueness of the Cauchy problem see~\cite[Thm~3.2.11]{Bar}). The advanced-minus-retarded solution operator is defined to be $\bE:=\bE^{-}-\bE^{+}$. Analogous properties hold for all normally hyperbolic operators.
Any solution to \eqref{eqn:ddeqnmotion} may be written in terms of $\bE$ as the next lemma shows.
\begin{lem}
\label{lem:solEf}
Any $\bgamma \in \mathscr{T}(M)$ solving $P(\bgamma)=0$  may be written as $\bgamma = \bE \bfab$ with $\bfab \in C^{\infty}_{0}(S^{0}_{2}(M))$.
\end{lem}
\begin{proof} See \cite[Thm~3.4.7]{Bar}. \end{proof}

The preceding lemma concerns general solutions to the hyperbolic equation \eqref{eqn:ddeqnmotion}, however, as we are looking for solutions to linearized gravity we are interested more particularly in
those solutions satisfying the de Donder condition, $\nabla\cdot\overline{\bgamma}=0$. Before showing when such a solution will be de Donder, we first require the following lemma.
\begin{lem}
\label{lem:overlinegammaEf}
For all $\bfab \in C_0^{\infty}(T^{0}_{2}(M))$, we have $
\overline{\bE \bfab}=\bE \overline{\bfab}$.
\end{lem}
\begin{proof} $\tilde{\bgamma}^{\pm}=\bE^{\pm} \overline{\bfab}$ are the unique solutions to $P(\tilde{\bgamma}^{\pm})=\overline{\bfab}$ with support in $J^{\pm}(\supp \overline{\bfab})$  and $\bgamma^{\pm}=\bE^{\pm} \bfab$ are the unique solutions to $P(\bgamma^{\pm})=\bfab$ with support in $J^{\pm}(\supp \bfab)$. Since trace-reversal commutes with $P$ (Lemma~\ref{lem:traceoperatorcommute}) we have
\begin{equation}
P(\bE^{\pm}\overline{\bfab})=\overline{\bfab}=\overline{P(\bE^{\pm} \bfab)}=P(\overline{\bE^{\pm}\bfab}).
\end{equation}
For any $\bk \in C^{\infty}_{0}(T^{0}_{2}(M))$ one may show that  $\supp \bk=\supp \overline{\bk}$ and hence $J^{\pm}(\supp\bk)=J^{\pm}(\supp \overline{\bk})$. Therefore by uniqueness (from the support properties)  $\bE^{\pm} \overline{\bfab}=\overline{\bE^{\pm} \bfab}$ and so $\bE \overline{\bfab}=\overline{\bE \bfab}$. \end{proof}

\begin{thm}
\label{thm:fundsolnsDD}For any $\bfab\in C^\infty_0(S^0_2(M))$, 
we have $\bE\bfab\in \mathscr{S}^{dD}(M)$ if and only if $\nabla\cdot\overline{\bfab}\in (\Box + \Lambda)C^{\infty}_{0}(T^{0}_{1}(M))$.
\end{thm}
\begin{proof} By definition, $\bE\bfab\in \mathscr{S}^{dD}(M)$ if and only if $\nabla \cdot \overline{\bE\bfab} \equiv 0$ or equivalently, using Lemma~\ref{lem:overlinegammaEf}, $\nabla \cdot \bE \overline{\bfab}=0$. 
Taking the divergence of $P (\bE^{\pm} \overline{\bfab})=\overline{\bfab}$ and utilising Lemma~\ref{lem:nablaP},
we find that
\begin{equation}
(\Box+\Lambda)(\nabla \cdot \bE^{\pm} \overline{\bfab})=\nabla \cdot\overline{\bfab}
\end{equation}
and deduce that $\nabla \cdot \bE^{\pm} \overline{\bfab}=\hat{\bE}^{\pm}\nabla\cdot \overline{\bfab}$, where $\hat{\bE}^{\pm}$ are the advanced and retarded Green's operators for $(\Box+\Lambda)$ on covector fields. Hence $\hat{\bE}\nabla\cdot \overline{\bfab}=\nabla \cdot \bE \overline{\bfab}=0$ and by \cite[Thm~3.4.7]{Bar}, this holds if and only if $\nabla\cdot\overline{\bfab} \in (\Box+\Lambda)C^{\infty}_{0}(T^{0}_{1}(M))$. \end{proof}

We now prove an identity concerning the action of $P$ on a pure gauge perturbation and then, using this, prove the relationship between $\bE$ and the Lie-derivative.
\begin{lem}
\label{lem:PLiederiv}
Given a $\bw \in C^{\infty}(T^{1}_{0}(M))$ on a cosmological vacuum background spacetime, then
\begin{equation}
\label{eqn:liederivhyperbolic}
\pounds_{(\Box+\Lambda)\bw}\bg=P(\pounds_{\bw}\bg).
\end{equation}
\end{lem}
\begin{proof}
Expanding out the left-hand side of \eqref{eqn:liederivhyperbolic} gives
\begin{equation}\label{eqn:liederivBox}
(\pounds_{(\Box+\Lambda)\bw}\bg)_{ab}
= 2\nabla_{(a} \Box w_{b)}+2\Lambda \nabla_{(a} w_{b)}.
\end{equation}
One can show, using the Riemann tensor identity \eqref{eqn:Riemann}, the Leibniz rule and \eqref{eqn:Riccitenscalvalue}, that 
\begin{equation}
\nabla_{a}\Box w_{b}=\Box(\nabla_{a}w_{b})-\Lambda\nabla_{a}w_{b} +2 R_{a \hspace{0.1cm}b}^{\hspace{0.1cm}c\hspace{0.1cm}d}\nabla_{c}w_{d}+w_{d}\nabla^{c}R_{acb}^{\phantom{acb}d}.
\end{equation}
We know from the proof of Lemma~\ref{lem:nablaP} that $\nabla_{a}R_{bcd}^{\phantom{bcd}a}=0$
on cosmological vacuum background spacetimes and therefore we have
\begin{equation}
\nabla_{a}\Box w_{b}=\Box(\nabla_{a}w_{b})-\Lambda\nabla_{a}w_{b} +2 R_{a \phantom{c}b}^{\phantom{a}c\phantom{b}d}\nabla_{c}w_{d}.
\end{equation}
Combining this with \eqref{eqn:liederivBox} gives the final result. \end{proof}

\begin{lem}
\label{lem:liederivE}
Given a $\bv \in C^{\infty}_{0}(T^{1}_{0}(M))$  on a cosmological vacuum background spacetime, then
\begin{equation}
\pounds_{\tilde{\bE}\bv}\bg=\bE(\pounds_{\bv}\bg),
\end{equation}
where $\tilde{\bE}^{\pm}$ are the advanced and retarded Green's operators for $(\Box+\Lambda)$ on vector fields.
\end{lem}
\begin{proof} Using Lemma~\ref{lem:PLiederiv}, 
$P(\pounds_{\tilde{\bE}^{\pm}\bv}\bg)
= \pounds_{(\Box+\Lambda)\tilde{\bE}^{\pm}\bv}\bg = \pounds_{\bv}\bg$. 
Thus $\pounds_{\tilde{\bE}^{\pm}\bv}\bg =\bE^{\pm} \pounds_{\bv}\bg$ 
and the result follows by uniqueness of solutions with past/future-compact support. \end{proof}

We are now able to prove the main result of this subsection.
\begin{thm}\label{thm:solutionspace}
Any  $\bgamma \in \mathscr{S}(M)$ is gauge equivalent to a $\bE\bfab$ for some $\bfab \in C^{\infty}_{0}(S^{0}_{2}(M))$ satisfying $\nabla\cdot\overline{\bfab}=0$.
\end{thm}
\begin{proof}  Combining Corollary~\ref{cor:deDsplitS}, Lemma~\ref{lem:solEf} and Theorem~\ref{thm:fundsolnsDD} gives $\bgamma \thicksim \bE \tilde{\bfab}$ with $\nabla\cdot\overline{\tilde{\bfab}}=(\Box+\Lambda)\bv^{\flat}$ for some $\bv \in C^{\infty}_{0}(T^{1}_{0}(M))$. Thus we also have
$\bgamma \thicksim \bE \tilde{\bfab}-\pounds_{\tilde{\bE}\bv}\bg =
\bE (\tilde{\bfab}-\pounds_{\bv}\bg)$, by Lemma~\ref{lem:liederivE}. 
Set $\bfab:=\tilde{\bfab}-\pounds_{\bv}\bg$, which is smooth and compactly supported on $M$. Calculating the divergence of the trace-reversal  of $\bfab$ gives
\begin{equation}
\nabla\cdot\overline{\bfab}=\nabla\cdot\overline{\tilde{\bfab}}-\nabla\cdot(\overline{\pounds_{\bv}\bg}) = (\Box+\Lambda)\bv^{\flat}-(\Box+\Lambda)\bv^{\flat}=0
\end{equation}
using Lemma~\ref{lem:divtracereversegauge}. \end{proof}

The remaining lemmas of this subsection are used in section~\ref{sec:quantization}. We will require the notion of past/future compactness of a subset of spacetime. A subset $S \subset M$ is said to be past/future compact if $J^{-}(p) \cap S$ or $J^{+}(p) \cap S$ is compact for all $p \in M$.
\begin{lem}
\label{lem:Efpastfuture}Given $\bfab \in C^{\infty}_{0}(S^{0}_{2}(M))$, if $\bgamma \in \mathscr{T}(M)$ solves $L_{ab}(\bgamma)=f_{ab}$ with $\supp\hspace{0.1cm}\bgamma$ compact to the past/future then  $\bgamma \thicksim -2\bE^{\pm}\overline{\bfab}$.
\end{lem}
\begin{proof} By Theorem~\ref{thm:deDsplit} (and the remark thereafter) there exists a
$\bw\in C^\infty(T^1_0(M))$ with support compact to the past/future such that  $\bgamma^{\prime}=\bgamma+\pounds_{\bw}\bg$ obeys  $\nabla \cdot \overline{\bgamma}^{\prime}=0$. Using that $L_{ab}(\pounds_{\bw}\bg)=0$ we have
$L_{ab}(\bgamma^{\prime})=L_{ab}( \bgamma) = f_{ab}$, which simplifies, on account of the de Donder condition in conjunction with Theorem~\ref{thm:Lident}, to 
$P(\overline{\bgamma}^{\prime}) = -2 \bfab$ 
. The solutions, to this inhomogeneous equation, with past/future compact support are $\overline{\bgamma}^{\prime}=-2\bE^{\pm}\bfab$. Lemma~\ref{lem:overlinegammaEf} entails that $\bgamma^{\prime}=-2\bE^{\pm}\overline{\bfab}$. Undoing the gauge transformation gives the required result.
\end{proof}

\begin{lem}
\label{lem:fequalsLh} Given a $\bfab \in C^{\infty}_{0}(S^0_2 (M))$ satisfying $\nabla \cdot \bfab=0$, suppose that $\bE \overline{\bfab}=\bE \pounds_{\bv}\bg$ for some $\bv \in C^{\infty}_{0}(T^{1}_{0}(M))$. Then
there exists $\bh \in C^{\infty}_{0}(S^{0}_{2}(M))$ such that
\begin{equation}
\bfab = -2 L(\bh).
\end{equation}
\end{lem}
\begin{proof} $\bE(\overline{\bfab}-\pounds_{\bv}\bg)=0$ and so by \cite[Thm~3.4.7]{Bar}
\begin{equation}
\label{eqn:fgaugeP}
\overline{\bfab} =\pounds_\bv\bg+  P(\bh)
\end{equation}
for some $\bh\in C_0^\infty(S_2^0(M))$. But for divergence-free $\bfab$, an application of Lemmas~\ref{lem:divtracereversegauge} and~\ref{lem:nablaP} gives $(\Box+\Lambda)(\bv^{\flat} + \nabla\cdot\overline{\bh})=0$ and hence
$\bv^{\flat} = -\nabla\cdot\overline{\bh}$. Reinserting this in \eqref{eqn:fgaugeP}, trace-reversing and applying
Theorem~\ref{thm:Lident} gives the result. \end{proof}

\section{Phase Space and Quantization}
\label{sec:phasespace}

\subsection{Phase Space}

We now construct the (complexified) phase space for linearized gravity on cosmological vacuum background spacetimes. Initially we consider the space $\mathscr{S}(M)$, which by the results of \cite{LeeWald} (equation~(2.21) and onwards in that reference, though our conventions differ) applied to the Lagrangian \eqref{eqn:lagrangian}, can be endowed with a complex-bilinear pre-symplectic product, whose action on perturbations $\bgamma^{1},\bgamma^{2} \in \mathscr{S}(M)$ is 
\begin{equation}
\label{eqn:presymp}
\omega_{\Sigma}(\bgamma^{1},\bgamma^{2}) = \int_{\Sigma}{(
\gamma^{1}_{ab}\pi^{ab}_{2} - \gamma^{2}_{ab}\pi^{ab}_{1})dvol_{\bh}},
\end{equation}
where $\Sigma$ is a spacelike Cauchy surface with future-pointing unit normal vector $\bn$,  $dvol_{\bh}$ denotes the volume element on $\Sigma$ associated with the induced spatial metric $\bh$ and $\bpi$ is defined in terms of the covariant conjugate momentum $\bPi$, given in \eqref{eqn:momentum}, by
\begin{equation}
\label{eqn:piab}
\pi^{ab}:=-n_{c}\Pi^{cab}.
\end{equation}
(Note that $n_c$ is {\em past}-pointing as a covector, owing to our signature convention.)
The product \eqref{eqn:presymp} is independent of the choice of Cauchy surface.
\begin{lem}
\label{lem:conserved}
Given $\bgamma^{1},\bgamma^{2} \in \mathscr{S}(M)$ and two spacelike Cauchy surfaces $\Sigma,\Sigma^{\prime}$ then $\omega_{\Sigma}(\bgamma^{1},\bgamma^{2})=\omega_{\Sigma^{\prime}}(\bgamma^{1},\bgamma^{2})$.
\end{lem}
\begin{proof} Defining the current of  $\bgamma^{1}$ and $\bgamma^{2}$ to be
$j^{c}(\bgamma^{1},\bgamma^{2}):=
\gamma^{2}_{ab}\Pi^{cab}_{1}-\gamma^{1}_{ab}\Pi^{cab}_{2}$,
the pre-symplectic product of these perturbations is thus
\begin{equation}
\omega_{\Sigma}(\bgamma^{1},\bgamma^{2}) = \int_{\Sigma}{n_{c}j^{c}(\bgamma^{1},\bgamma^{2})dvol_{\bh}}.
\end{equation}
Now, the divergence of the current is $
\nabla_{c}j^{c}=\gamma^{2}_{ab}L^{ab}(\bgamma^{1})-\gamma^{1}_{ab}L^{ab}(\bgamma^{2})=0$,
where we have used \eqref{eqn:Eulerlagrange} and symmetry properties of $S^{abcd}$ and $T^{abcdef}$. Using the divergence theorem over the region bounded by the two Cauchy surfaces $\Sigma,\Sigma^{\prime}$ gives the desired result. \end{proof}

Due to the preceding lemma, the $\Sigma$ will be dropped from $\omega_\Sigma$ from this point on if we are dealing purely with solutions.

To make the pre-symplectic product into a symplectic product, it is
necessary to account for the degeneracies of \eqref{eqn:presymp}, that is, non-trivial solutions whose pre-symplectic product with all solutions is zero. The subspace of degeneracies is also known as the {\em radical} of the pre-symplectic form $\omega$. The next lemma shows that $\omega$ is gauge invariant and therefore (as is well-known) pure gauge solutions are degeneracies on the space of solutions.
Here, we work with the broader class of pure gauge solutions  $\hat{\mathscr{G}}(M)$ defined in~\eqref{eq:Ghatdef}; that is, those spacelike-compact pure
gauge solutions induced by arbitrary smooth vector fields.  Recall that 
$\hat{\mathscr{G}}(M)$ coincides with $\mathscr{G}(M)$ if $M$ has
compact Cauchy surfaces, but the two can (but do not always) differ in the non-compact case.

\begin{lem}
\label{lem:puregaugedegen}
$\hat{\mathscr{G}}(M)$ is contained in the radical of $\omega$.
\end{lem}
\begin{proof} Suppose $\bw\in C^\infty(T^1_0(M))$ and $\bgamma\in\mathscr{S}(M)$, and let $\Sigma$ be a smooth spacelike Cauchy surface. From Theorem~\ref{thm:C} we have the identity
\begin{equation}
\omega_{\Sigma}(\bgamma, \pounds_{\bw}\bg) = 2\int_{\Sigma}{w^{b}C^{\Sigma}_{b}(\Data_{\Sigma}(\bgamma))dvol_{\bh}}
\end{equation}
and the right-hand side vanishes because $\bC^{\Sigma}(\Data_{\Sigma}(\bgamma))=0$. \end{proof}

At least in the case that $M$ has compact Cauchy surfaces, we
may prove that this exhausts the space of degeneracies of~\eqref{eqn:presymp}. It is natural to conjecture that the same is true for a large class of background spacetimes with non-compact Cauchy surfaces as well. 
\begin{thm}\label{thm:radical}
If $M$ has compact Cauchy surfaces, 
the radical of $\omega$ is precisely the subspace of pure gauge solutions $\hat{\mathscr{G}}(M)$ (which coincides with $\mathscr{G}(M)$ in this case).
That is, given $\bgamma^{\prime} \in \mathscr{S}(M)$ such that $\omega (\bgamma^{\prime},\bgamma)=0$ for all $\bgamma \in \mathscr{S}(M)$, then $\bgamma^{\prime} \in \hat{\mathscr{G}}(M)$.
\end{thm}
The proof requires results of Moncrief on the ADM formulation, and is given in Appendix~\ref{sec:ADMnondegen}. 

In any spacetime for which $\hat{\mathscr{G}}(M)$ is the radical of $\omega$,
we obtain the complexified phase space as the quotient space
\begin{equation}\label{eqn:PM_def}
\mathscr{P}(M):=\mathscr{S}(M)/\hat{\mathscr{G}}(M)
\end{equation}
with weakly non-degenerate symplectic product
\begin{equation}
\label{eqn:sympprod}
\omega([\bgamma^{1}],[\bgamma^{2}])=\int_{\Sigma}{(\gamma^{1}_{ab}\pi^{ab}_{2} - \gamma^{2}_{ab}\pi^{ab}_{1})dvol_{\bh}}.
\end{equation}
As this is independent of the choice of representative we may choose de Donder representatives $\tilde{\bgamma}^{i}$ ($i=1,2$) for each class, 
for which the associated momenta are $\tilde{\pi}^{ab}_{i}=
n_c \mathscr{D}^{cab}[\overline{\tilde{\gamma}}_{i}]$, where the 
differential operator 
\[
\mathscr{D}^{cab}[\gamma] = 
\frac{1}{2}\nabla^{c}\gamma^{ab}
-\frac{1}{2}\nabla^{b}\gamma^{ca}
-\frac{1}{2}\nabla^{a}\gamma^{cb}
\]
has the property $\nabla_c \mathscr{D}^{cab}[\bgamma] = 
\frac{1}{2}(P\bgamma)^{ab}- \Lambda\gamma^{ab}$ for de Donder $\bgamma$
on cosmological background spacetimes. 
%

Under complex conjugation, $\omega([\bgamma^{1}],[\bgamma^{2}])^* = 
\omega([\bgamma^{1*}],[\bgamma^{2*}])$.
The real phase space $\mathscr{P}_{\mathbb{R}}(M)$ is obtained by restricting all the above definitions
to real-valued solutions and real-valued gauge transformations. 

We now wish to find out what the symplectic product is in terms of a de Donder resprentative written in terms of the solution operator, as in Theorem~\ref{thm:fundsolnsDD}.
\begin{thm}
\label{thm:sympprodsolDD}
Given $\bgamma \in \mathscr{S}(M)$ and $\bfab \in C^{\infty}_{0}(S^{0}_{2}(M))$ satisfying $\nabla\cdot \bfab=0$, then
\begin{equation}
\label{eqn:symplecticEfgamDD}
\omega([\bE\overline{\bfab}],[\bgamma])=
-\frac{1}{2}\int_{M}{\gamma^{ dD}_{ab}f^{ab}dvol_{\bg}},
\end{equation}
where $\bgamma^{dD}$ denotes a de Donder representative of $[\bgamma]$.
\end{thm}
\begin{proof} 
By Theorem~\ref{thm:fundsolnsDD}, $\bE\overline{\bfab}$ is a de Donder solution (and its proof
shows that $\bE^\pm\overline{\bfab}$ also obey the de Donder condition). If we select a de Donder representative $\bgamma^{dD}$ of $[\bgamma]$ then the left-hand side of \eqref{eqn:symplecticEfgamDD} may be written as
\begin{equation}\label{eqn:sympprod1}
\omega([\bE\overline{\bfab}],[\bgamma]) = \int_{\Sigma}
\left((\bE\overline{\bfab})_{ab}\mathscr{D}^{cab}[\overline{\bgamma^{dD}}]
- \gamma^{dD}_{ab} \mathscr{D}^{cab}[\bE\bfab]\right)n_c dvol_{\bh}
\end{equation}
using that $\overline{\bE\overline{\bfab}}=\bE\overline{\overline{\bfab}}=\bE \bfab$ from Lemma~\ref{lem:overlinegammaEf}.

As $\supp\bfab$ is compact we may choose Cauchy surfaces $\Sigma,\Sigma^{\prime}$ such that $\Sigma \subset I^{+}(\Sigma^{\prime})$ and $\supp \bfab \subset I^{+}(\Sigma^{\prime})\cap I^{-}(\Sigma)$. The region bounded by these two Cauchy surfaces is henceforth denoted by $V$. We will utilise the Gauss Theorem applied to the vector field
\begin{equation}\label{eqn:sympprod2}
v^c = \gamma^{dD}_{ab} \mathscr{D}^{cab}[\bE^+\bfab] - 
(\bE^+\overline{\bfab})_{ab}\mathscr{D}^{cab}[\overline{\bgamma^{dD}}]
\end{equation}
on the region $V$ to prove the desired result. Using the formula 
$\nabla_c \mathscr{D}^{cab}[\bgamma] = 
\frac{1}{2}P(\bgamma)^{ab}- \Lambda\gamma^{ab}$,
applied to the de Donder perturbations $\bE^+\overline{\bfab}$ and $\bgamma^{dD}$,
the divergence of $\bv$ is calculated to be
\begin{equation}
\nabla\cdot\bv = \frac{1}{2}\gamma^{dD}_{ab}f^{ab} 
+ (\nabla_c\gamma_{ab}^{dD}) \mathscr{D}^{cab}[\bE^+\bfab] -
(\nabla_c(\bE^+\overline{\bfab})_{ab}) \mathscr{D}^{cab}[\overline{\bgamma^{dD}}]
\end{equation}
where we have also used the fact that $P(\bgamma^{dD}) = 0$ and $P(\bE^{+}\bfab)=\bfab$. Using
the de Donder condition again, 
the second and third terms may be seen to cancel and we have 
$\nabla\cdot\bv =\frac{1}{2} \gamma^{dD}_{ab}f^{ab}$.
By Gauss' Theorem applied to the region V, where $\partial V=\Sigma \cup \Sigma^{\prime}$ and with $\bn$ denoting the future-pointing unit normal vector on $\partial V$, this gives
\begin{equation}
\frac{1}{2}\int_{V}{\gamma^{dD}_{ab}f^{ab}dvol_{\bg}} =- \int_{\Sigma}{n_{c}v^{c} dvol_{\bh}} + \int_{\Sigma^{\prime}}{n_{c}v^{c} dvol_{\bh}}
\end{equation}
The integral over $\Sigma^{\prime}$ is zero because $\bE^{+}\bfab$ and its derivative vanish on $\Sigma^{\prime}$, so we obtain 
\begin{equation}
\frac{1}{2}\int_{V}{\gamma^{dD}_{ab}f^{ab}dvol_{\bg}} =- \int_{\Sigma} n_{c}\left( \gamma^{dD}_{ab} \mathscr{D}^{cab}[\bE^+\bfab] - 
(\bE^+\overline{\bfab})_{ab}\mathscr{D}^{cab}[\overline{\bgamma^{dD}}]
\right)dvol_{\bh}.
\end{equation}
As $\bE^{-}\overline{\bfab}$ and its derivative vanish at $\Sigma$ we may replace $\bE^{+}$ by $-\bE$ and use \eqref{eqn:sympprod1} to obtain the final result. \end{proof}

\subsection{Observables}
\label{sec:observables}

Observables are functions on the (complexified) phase space, $\mathscr{P}(M)$. As for the scalar~\cite{DimockScalar} and electromagnetic fields~\cite{DimockEM}, the
observables that will form the basis for the quantum theory will be 
certain smeared fields. In our case, we wish to consider
integrals of the form $\int{\gamma_{ab}f^{ab} dvol_{\bg}}$, where
$\bgamma\in\mathscr{S}(M)$ and $\bfab\in C_0^\infty(T^0_2 (M))$;
however, this will only be gauge invariant, that is, independent of the choice of representative of the equivalence class of $\bgamma$, and thus a well-defined
function on $\mathscr{P}(M)$, if $\bfab$ is restricted. 

\begin{lem}
\label{lem:gaugeinvobs}
For $\bfab\in C_0^\infty(T^0_2(M))$, we have $\int_{M}\gamma_{ab} f^{ab} dvol_{\bg}=0$ for all $\bgamma\in\hat{\mathscr{G}}(M)$ if and only if $\nabla^{a}f_{(ab)}=0$.
\end{lem}
\begin{proof} For $\bgamma=\pounds_\bw\bg\in\hat{\mathscr{G}}(M)$, we 
have 
\begin{align}
\label{eqn:gaugeFf}
\int_{M}{(\nabla_{(a}w_{b)}) f^{ab} dvol_{\bg}} &=\int_{M}{\nabla_{a}(w_{b}f^{(ab)}) dvol_{\bg}} - \int_{M}{w_{b}(\nabla_{a}f^{(ab)} )dvol_{\bg}}
\nonumber \\
&=
- \int_{M}{w^{b}(\nabla^{a}f_{(ab)} )dvol_{\bg}},
\end{align}
where we moved the symmetrization to $\bfab$ and then used the Leibniz rule and the divergence theorem together with the support properties of $\bfab$ to obtain \eqref{eqn:gaugeFf}. For \eqref{eqn:gaugeFf} to 
vanish it is clearly sufficient that $\nabla^{a}f_{(ab)}=0$; as $\bw$ may, in particular, be any element of $C_0^\infty(T^1_0(M))$, necessity holds as well. \end{proof}

We thus arrive at the final definition of observables for our theory. 
\begin{dfn}
\label{dfn:observables}
For each $\bfab \in C^{\infty}_{0}(T^{0}_{2}(M))$ satisfying $\nabla^{a}f_{(ab)}=0$ the observable $F_{\bfab}:\mathscr{P}(M)\to\mathbb{C}$ is given by
\begin{equation}
\label{eqn:observables}
F_{\bfab}([\bgamma])=\int{\gamma_{ab}f^{ab} dvol_{\bg}}
\end{equation}
(and is necessarily gauge invariant).
\end{dfn}

The observables \eqref{eqn:observables} satisfy four important relations. The reader who is
familiar with the algebraic formulation of the real scalar field might expect that we would
only state (iv) for compactly supported $\bfab$; the reason for our use of $\bfab$ with time-compact support will become clear in the next subsection. For (iv), recall that $L$ is defined on arbitrary tensors by (6); in particular it 
vanishes on antisymmetric tensors.
\begin{thm}
\label{thm:relations}
Given any $[\bgamma]\in\mathscr{P}(M)$, the $F_{\bfab}$'s satisfy:
\begin{tightenumerate}
\item {\em Complex linearity:} $ F_{\alpha \bfab+\beta \tilde{\bfab}}([\bgamma])=\alpha F_{\bfab}([\bgamma])+\beta F_{\tilde{\bfab}}([\bgamma])$ for all $\alpha,\beta \in \mathbb{C}$ and all $\bfab,\bfab^{\prime} \in C^{\infty}_{0}(T^{0}_{2}(M))$ satisfying $ \nabla^{a}f_{(ab)}=0=\nabla^{a}f^{\prime}_{ (ab)}$;
\item {\em Hermiticity:} $ F_{\bfab}([\bgamma])^{*}= F_{\bfab^{*}}([\bgamma^*])$ for all $\bfab \in C^{\infty}_{0}(T^{0}_{2}(M))$ satisfying $\nabla^{a}f_{(ab)}=0$;
\item {\em Symmetry:} $ F_{\bfab}([\bgamma])=0$ for all antisymmetric  $\bfab \in C^{\infty}_{0}(T^{0}_{2}(M))$;
\item {\em Field equation \eqref{eqn:nonhypeqn} holds:} $F_{L(\bfab)}([\bgamma])=0$ for all $\bfab \in C^{\infty}_{TC}(T^{0}_{2}(M))$ with $L(\bfab)\in C_0^\infty(S^0_2(M))$.
\end{tightenumerate}
\end{thm}
\begin{proof} (i), (ii), (iii)  are obvious, and (iv) holds because $L_{ab}$ is formally self adjoint. \end{proof}

We can now give a variant of Theorem~\ref{thm:sympprodsolDD} in which there is no longer any need to work with de Donder representatives. The theorem is proved on the assumption that $\bfab$ is symmetric and has vanishing divergence; of course any $\bfab$ whose symmetric part is divergenceless may be decomposed into symmetric and antisymmetric parts and the antisymmetric part does not
contribute to $F_\bfab([\bgamma])$ by part~(iii) of Theorem~\ref{thm:relations}.
\begin{thm}
\label{thm:sympprodEFgammafinal}
Given $[\bgamma] \in \mathscr{P}(M)$ and $\bfab \in C^{\infty}_{0}(S^{0}_{2}(M))$ satisfying $\nabla^{a}f_{ab}=0$, then
\begin{equation}
\label{eqn:symplecticEfgam}
 F_{\bfab}([\bgamma])
=\int_{M}{\gamma_{ab}f^{ab} dvol_{\bg}} = -2\omega([\bE \overline{\bfab}],[\bgamma]) .
\end{equation}
\end{thm}
\begin{proof} As $\nabla^{a}f_{ab}=0$ we can use Theorem~\ref{thm:sympprodsolDD} to give
\begin{equation}
\omega([\bE \overline{\bfab}],[\bgamma])=-\frac{1}{2}\int_{M}{\gamma^{dD}_{ab}f^{ab} dvol_{\bg}}.
\end{equation}
As $\bfab$ satisfies the requirements of Lemma~\ref{lem:gaugeinvobs} we may replace $\bgamma^{dD}$ by $\bgamma$ and the result follows. \end{proof}

The final two results of this section make use of the weak non-degeneracy
of the symplectic product, in contrast to those above. First,
we show that there are sufficiently many observables to distinguish points of $\mathscr{P}(M)$; by this we mean that given two distinct equivalence classes of solutions $[\bgamma]$ and $[\bgamma^{\prime}]$ then there exists at least one $\bfab \in C^{\infty}_{0}(T^{0}_{2}(M))$ satisfying $\nabla^{a}f_{(ab)}=0$ such that $F_{\bfab}([\bgamma])\neq F_{\bfab}([\bgamma^{\prime}])$. 
We are grateful to Atsushi Higuchi for raising this question with us.
\begin{thm} \label{thm:D} 
Assuming weak non-degeneracy holds [thus, in particular, for any $M$ with
compact Cauchy surfaces] then, for any distinct $[\bgamma_1],[\bgamma_2]\in \mathscr{P}(M)$, there exists a $\bfab \in C^{\infty}_{0}(S^{0}_{2}(M))$ with $\nabla^a f_{ab}=0$ such that $F_{\bfab}([\bgamma_1])\neq
F_{\bfab}([\bgamma_2])$. 
\end{thm}
\begin{proof}  By weak non-degeneracy there exists a $[\bgamma] \in \mathscr{P}(M)$ such that 
\begin{equation}
\label{eqn:sympprodneq}
\omega([\bgamma],[\bgamma_1]) \neq \omega([\bgamma],[\bgamma_2]).
\end{equation}
By Theorem~\ref{thm:solutionspace}, $[\bgamma]=[\bE \overline{\bfab}]$ for some $\bfab \in C^{\infty}_{0}(S^{0}_{2}(M))$ satisfying $\nabla \cdot \bfab =0$. Using Theorem~\ref{thm:sympprodEFgammafinal} together with \eqref{eqn:sympprodneq} gives
\begin{equation}
F_{\bfab}([\bgamma_{1}])=-2\omega([\bE \overline{\bfab}],[\bgamma_1]) \neq -2\omega([\bE \overline{\bfab}],[\bgamma_2]) = F_{\bfab}([\bgamma_{2}]).
\end{equation}
\end{proof}

Finally, we compute the Poisson bracket of two observables in our class. Here, we regard $\mathscr{P}(M)$ as an infinite-dimensional symplectic manifold, with the smooth structure determined (as a Fr\"{o}licher space---see, e.g.~\cite[Ch.~23]{Globalanalysis})
by the symplectic form $\omega$. Thus a curve $c:\mathbb{R}\to\mathscr{P}(M)$ is defined to be smooth if $t\mapsto \omega(v,c(t))$ is smooth for all $v$ and  a function $F:\mathscr{P}(M)\to \mathbb{C}$ is defined to be smooth if $F\circ c$ is smooth for every smooth  curve $c$; in particular, $\omega$ itself is a smooth function in each slot separately,
and therefore our $F_{\bfab}$'s are smooth. The Poisson bracket of two smooth functions
$F,G\in C^\infty(\mathscr{P}(M))$ is given in terms of the exterior derivatives of $F$ and $G$ by
\begin{equation}
\{F,G\}([\bgamma])=dF(IdG)|_{[\bgamma]},
\end{equation}
where the Hamiltionian vector field $IdG$ induced by $G$ satisfies
\begin{equation}
\label{eqn:IdGeqn}
\omega_{[\bgamma]}(IdG|_{[\bgamma]},v)=dG|_{[\bgamma]}(v)
\end{equation}
for $v \in T_{[\bgamma]}\mathscr{P}(M)$ (we will show that this is uniquely defined in our context). Here $\omega_{[\bgamma]}$ is the symplectic form at 
$[\bgamma]\in\mathscr{P}(M)$. Under the identification $T_{[\bgamma]}\mathscr{P}(M) \cong \mathscr{P}(M)$,  $\omega_{[\bgamma]}$ is replaced by $\omega$. 
\begin{thm}\label{thm:Poisson}
Assuming weak non-degeneracy holds [thus, in particular, if $M$ has
compact Cauchy surfaces], the Poisson bracket of two observables satisfying Definition~\ref{dfn:observables} is 
\begin{equation}
\label{eqn:Poisson}
\{ F_{\bfab}, F_{\bfab^{\prime}}\} = -2 \bE(\bfab^{s},\overline{\bfab^{\prime s}})
= 4\omega([\bE\overline{\bfab^{s}}],[\bE\overline{\bfab^{\prime s}}]),
\end{equation}
where $\bfab^{s}$ denotes the symmetric part of $\bfab$ [i.e., $f^s_{ab}=  f_{(ab)}$] and
\begin{equation}
\label{eqn:bidistribE}
\bE(\bfab^{s},\overline{\bfab^{\prime s}}) :=\int_{M}{f^{(ab)} (E^{\hspace{0.2cm}cd}_{ab}\overline{f}^{\prime}_{(cd)}) dvol_{\bg}}.
\end{equation}
\end{thm}
\begin{proof} We note that $dF_{\bfab}|_{[\bgamma]}([\bgamma^{\prime}])=F_{\bfab}([\bgamma^{\prime}])$ by linearity of $F_{\bfab}$. Thus, upon using \eqref{eqn:IdGeqn} and then Theorem~\ref{thm:sympprodEFgammafinal} we have
\begin{equation}
\omega(IdF_\bfab|_{[\bgamma]},[\bgamma]) = F_{\bfab}([\bgamma]) 
= F_{\bfab^s}([\bgamma]) = -2\omega([\bE \overline{\bfab^{s}}],[\bgamma]).
\end{equation}
 By weak non-degeneracy, this gives
$IdF_\bfab =-2 [\bE \overline{\bfab^{s}}]$ and thus
\begin{align}
\{F_{\bfab},F_{\bfab^{\prime}}\}([\bgamma]) &= - dF_{\bfab}|_{[\bgamma]}(2[\bE \overline{\bfab^{\prime s}}]) = - 2F_\bfab([\bE \overline{\bfab^{\prime s}}]) 
=- 2F_{\bfab^s}([\bE \overline{\bfab^{\prime s}}]) \notag\\
&
=  -2 \bE(\bfab^{s},\overline{\bfab^{\prime s}}) .
\end{align}
Finally, we use Theorem~\ref{thm:sympprodEFgammafinal} to obtain the final equality in \eqref{eqn:Poisson} . \end{proof}

Although superficially they do not appear the same, in fact the propagator $2\bE(\bfab^s,\overline{\bfab^{\prime s}})$ is the same as the one considered by Lichnerowicz in equation~(21.3) of \cite{Lichnerowicz}. If one expands out the trace-reversal then
\begin{equation}
\{ F_{\bfab}, F_{\bfab^{\prime}}\}=-2\bE(\bfab^s,\overline{\bfab^{\prime s}})=-2\bE(\bfab^s,\bfab^{\prime s})+ E(f,f^{\prime}),
\end{equation}
where $f=f_a^{\phantom{a}a}$ denotes the trace of $\bfab$; note the appearance of the
scalar propagator $E$ in the last term of this equation. 

\subsection{Reformulation of the phase space}
\label{sec:phasespacereform}

In the next subsection we will apply Dirac quantization to the observables discussed above
to obtain a $*$-algebra of observables. An alternative approach would be to directly quantize the complexified symplectic space $\mathscr{P}(M)$ as an infinitesimal Weyl algebra. Similarly, the Weyl algebra provides another 
quantization based directly on the real symplectic space. (See, e.g., \cite{FewVer:dynloc2} for a presentation of both constructions emphasizing their functorial nature.) To clarify the relationship between these 
prescriptions and Dirac quantization, we now show that the observables form a complexified symplectic space, with their Poisson bracket as the symplectic product, that is symplectically isomorphic to $\mathscr{P}(M)$. 

Let $\mathscr{F}(M)$ be the set of $\bfab\in C^{\infty}_{0}(S^{0}_{2}(M))$ that are divergence-free, i.e., $\nabla^a f_{ab}=0$. Then, by Theorem~\ref{thm:relations}, $F:\bfab\mapsto F_{\bfab}$ is a linear map, intertwining the complex conjugations of $\mathscr{F}(M)$ and of complex-valued functions on $\mathscr{P}(M)$ and whose range coincides with the full set of observables in our class.
In cases where $\mathscr{P}(M)$ is weakly non-degenerate, it is clear from
Theorem~\ref{thm:sympprodEFgammafinal} that the kernel of the map $F$ is precisely the
subspace of $\bfab\in\mathscr{F}(M)$ for which $\bE\overline{\bfab}\in\hat{\mathscr{G}}(M)$. 
The next result will characterize this kernel in an attractively simple way. 

For this subsection only, we consider extensions of various Green's operators from smooth compactly supported tensor fields to tensor fields with time-compact support,
i.e., the support lies between two Cauchy surfaces. The retarded/advanced Green's operators are extended by
defining $\bE^\pm\bfab$ for $\bfab\in C^{\infty}_{TC}(T^{0}_{2}(M))$ to be the unique 
solution to $P (\bgamma)=\bfab$ that vanishes to the past/future of the support of $\bfab$,
and we define $\bE\bfab=\bE^-\bfab-\bE^+\bfab$. Many standard results have analogues for these extensions, that follow by the same arguments as the standard versions.
In particular, $P$ has trivial kernel in $C^\infty_{TC}(T^0_2(M))$, the kernel of the extended operator $\bE$ is precisely $P (C^\infty_{TC}(T^0_2(M)))$
and any smooth solution to $P(\bgamma)=0$ may be written in the form $\bE\bfab$ for some
$\bfab\in C^\infty_{TC}(T^0_2(M))$. These facts may be summarised in the commutative diagram
\[
\begin{tikzpicture}[baseline=0 em, description/.style={fill=white,inner sep=2pt}]
\matrix (m) [ampersand replacement=\&,matrix of math nodes, row sep=3em,
column sep=2.25em, text height=1.5ex, text depth=0.25ex]
{ 0 \& C^\infty_{TC}(T^0_2(M)) \&  C^\infty_{TC}(T^0_2(M)) 
\&   C^\infty(T^0_2(M))  \&  C^\infty(T^0_2(M)) \\ 
0 \& C^\infty_{0}(T^0_2(M)) \&  C^\infty_{0}(T^0_2(M)) 
\&   C^\infty_{SC}(T^0_2(M))  \&  C^\infty_{SC}(T^0_2(M))\\ };
\path[->]
(m-1-1) edge (m-1-2)
(m-1-2) edge node[above] {\scriptsize{$P$}} (m-1-3)
(m-1-3) edge node[auto] {\scriptsize{$\bE$}} (m-1-4)            
(m-1-4) edge node[above] {\scriptsize{$P$}} (m-1-5)
(m-2-1) edge (m-2-2)
(m-2-2) edge node[above] {\scriptsize{$P$}} (m-2-3)
            edge (m-1-2)
(m-2-3) edge node[auto] {\scriptsize{$\bE$}} (m-2-4)
            edge (m-1-3)
(m-2-4) edge node[above] {\scriptsize{$P$}} (m-2-5)
            edge (m-1-4)
(m-2-5) edge (m-1-5);
\end{tikzpicture}
\]
in which the upper and lower lines are exact sequences of linear maps and the vertical arrows
are the obvious inclusion maps. The lower line is the content of~\cite[Thm 3.4.7]{Bar}.

Similar results hold for the extension $\tilde{\bE}$ of the 
analogous vector Green's operators. Moreover, all the results from section~\ref{sec:greensoperators}, concerning the intertwining of the Green's operators with various other operations, 
in particular Lemma~\ref{lem:liederivE},  carry straight through and apply to the extensions.

The kernel of the map $F$ may now be determined. In the following, 
$\mathscr{L}(M)=\{L(\bk): \bk\in  C^\infty_{0}(S^0_2(M))\}$, while
$\hat{\mathscr{L}}(M)=\{L(\bk): \bk\in  C^\infty_{TC}(S^0_2(M))\}\cap C_0^\infty(S^0_2(M))$. 
A consequence of the next result is that $\hat{\mathscr{G}}(M)=\mathscr{G}(M)$ if and only if
$\hat{\mathscr{L}}(M)=\mathscr{L}(M)$.
\begin{lem}\label{lem:Ghat}
Suppose $\bfab\in\mathscr{F}(M)$. Then  $\bE\overline{\bfab}\in\hat{\mathscr{G}}(M)$ if and only if  $\bfab \in\hat{\mathscr{L}}(M)$; similarly, $\bE\overline{\bfab}\in\mathscr{G}(M)$ if and only if  $\bfab \in\mathscr{L}(M)$.
\end{lem}
\begin{proof} ($\Rightarrow$) Suppose $\bE\overline{\bfab}=\pounds_\bw\bg\in \hat{\mathscr{G}}(M)$. As $\nabla\cdot \bfab=0$, $\pounds_\bw\bg$ is a de Donder solution. In consequence, $\bw$ satisfies 
$(\Box+\Lambda)\bw=0$ and may be written as $\bw=\tilde{\bE}\bv$ for some $\bv \in C^{\infty}_{TC}(T^{1}_{0}(M))$. Therefore we have $\pounds_{\bw}\bg=\pounds_{\tilde{\bE}\bv}\bg=\bE \pounds_{\bv}\bg$, using (the analogue of) Lemma~\ref{lem:liederivE}. Thus $\bE\overline{\bfab} = \bE\pounds_{\bv}\bg$ and so
\begin{equation}
\label{eqn:gaugevfPF}
\pounds_{\bv}\bg=\overline{\bfab}+P(\bk)
\end{equation}
for some $\bk \in C^{\infty}_{TC}(S^{0}_{2}(M))$.  Taking the trace-reversal and then the divergence of this equation gives $(\Box+\Lambda)\bv^{\flat}=(\Box+\Lambda)\nabla \cdot \overline{\bk}$,
from which one can deduce, using time-compactness of $\bv$ and $\nabla \cdot \overline{\bk}$, that $\bv^{\flat}=\nabla \cdot \overline{\bk}$. Substituting this result into \eqref{eqn:gaugevfPF}  and using Theorem~\ref{thm:Lident} gives $\overline{\bfab} = \pounds_{(\nabla \cdot \overline{\bk})^{\sharp}}\bg - P (\bk) = 2 \overline{L(\bk)}$. Therefore $\bfab = L(2\bk)\in\hat{\mathscr{L}}(M)$  as required. 

($\Leftarrow$) Conversely, given  $\bk \in C^{\infty}_{TC}(S^{0}_{2}(M))$ satisfying $L(2\bk) =\overline{\bfab}$,
let $\bw=\tilde{\bE}(\nabla \cdot \overline{\bk})^{\sharp} \in C^{\infty}(T^{1}_{0}(M))$, which gives
$\pounds_{\bw}\bg = \pounds_{\tilde{\bE}(\nabla \cdot \overline{\bk})^{\sharp}}\bg = \bE \pounds_{(\nabla \cdot \overline{\bk})^{\sharp}}\bg$, using (the analogue of) Lemma~\ref{lem:liederivE} in the final equality. Now by Theorem~\ref{thm:Lident} this becomes $\pounds_{\bw}\bg = \bE (2 \overline{L(\bk)}+P(\bk)) = \bE\overline{\bfab}$, 
because $\bE P (\bk)=0$. As $\bfab$ is compactly supported by assumption, we deduce that
$\pounds_\bw\bg\in  C^{\infty}_{SC}(S^{0}_{2}(M))$ and hence $\bE\overline{\bfab}\in\hat{\mathscr{G}}(M)$. 

The second statement has an exactly analogous proof, replacing $C^\infty_{TC}$ by $C^\infty_0$ and $C^\infty$ by $C^\infty_{SC}$, and hatted spaces by their unhatted counterparts throughout. \end{proof}

It follows that the vector space of observables is isomorphic to $\mathscr{F}(M)/
\hat{\mathscr{L}}(M)$ if weak non-degeneracy holds. 

\begin{prp}\label{prop:phase_reform}
The map $\bfab+\hat{\mathscr{L}}(M)\mapsto [\bE\overline{\bfab}]$ is a linear 
isomorphism of $\mathscr{F}(M)/
\hat{\mathscr{L}}(M)$ and $\mathscr{P}(M)$. Accordingly, if weak non-degeneracy holds [thus, in particular, for any $M$ with
compact Cauchy surfaces],  there is an isomorphism of the 
space of observables and $\mathscr{P}(M)$ such that $F_{\bfab}\mapsto [\bE\overline{\bfab}]$
for all $\bfab\in\mathscr{F}(M)$, and which is symplectic if the observables are equipped with 
symplectic product $\sigma(F_{\bfab_1}, F_{\bfab_2}) = \frac{1}{4}\{F_{\bfab_1}, F_{\bfab_2}\}$. 
\end{prp}
\begin{proof} Lemma~\ref{lem:Ghat} shows that the given map is
well-defined and injective, while surjectivity is evident by Theorem~\ref{thm:solutionspace}. 
The symplectic product is seen to be preserved by combining  Theorems~\ref{thm:sympprodEFgammafinal} and~\ref{thm:Poisson} to give $4\omega([\bE\overline{\bfab}_1],[\bE\overline{\bfab}_2]) =\{F_{\bfab_1}, F_{\bfab_2}\}$
for any $\bfab_1,\bfab_2\in\mathscr{F}(M)$.
\end{proof}

\subsection{Quantization}
\label{sec:quantization}

To quantize the theory we follow Dirac's prescription, which requires that we seek operators $\hat{F}_{\bfab}$, labelled by test tensors $\bfab$ with divergence-free symmetric part and whose commutator is given by $[\hat{F}_{\bfab},\hat{F}_{\bfab^{\prime}}]=i\{F_{\bfab},F_{\bfab^{\prime}}\} \II$. We also expect these operators to respect the quantum analogues of the relations for the classical observables from Theorem~\ref{thm:relations}. As we
will regard these objects as smeared quantum fields, we use the notation  $[\bgamma](\bfab)$ 
rather than $\hat{F}_{\bfab}$, cautioning that the $[\bgamma]$ in such expressions is not to be confused with an equivalence class
of classical solutions. 
Combining the above requirements gives the algebraic relations:
\begin{tightenumerate}
\item {\em Complex-linearity}: $[\bgamma](\alpha \bfab_{1} + \beta \bfab_{2}) = \alpha [\bgamma](\bfab_{1})+\beta [\bgamma](\bfab_{2})$ for all $ \alpha,\beta \in \mathbb{C}$ and all $\bfab_{i} \in C^{\infty}_{0}(T^{0}_{2}(M))$ such that $\nabla^{a}(f_{i})_{(ab)}=0$;
\item {\em Hermiticity}: $[\bgamma](\bfab)^{*} = [\bgamma](\bfab^{*})$ for all $\bfab \in C^{\infty}_{0}(T^{0}_{2}(M))$ such that $\nabla^{a}f_{(ab)}=0$;
\item {\em Symmetry}: $[\bgamma](\bfab) = 0$ for all antisymmetric $ \bfab \in C^{\infty}_{0}(T^{0}_{2}(M))$;
\item {\em Field equation}: $[\bgamma](L(\bfab)) = 0 $ for all $ \bfab \in C^{\infty}_{TC}(T^{0}_{2}(M))$ such that $L(\bfab)\in C_0^\infty(S^0_2 (M))$.
\item {\em Commutation relation}: $[[\bgamma](\bfab_{1}),\ [\bgamma](\bfab_{2})] = -2i\bE(\bfab^{s}_{1},\overline{\bfab^{s}_{2}}) \II$ for all $\bfab_{i} \in C^{\infty}_{0}(T^{0}_{2}(M))$ such that $ \nabla^{a}(f_{i})_{(ab)}=0$.
\end{tightenumerate}
Note that (a) in the statement of the field equation, the linearized Einstein tensor $L(\bfab)$
is symmetric (by definition of $L$) and divergence free (by the Bianchi identities) for any 
$\bfab \in C^{\infty}_{TC}(T^{0}_{2}(M))$;
(b) the final relation implies that the commutator vanishes if the $\bfab_{i}$ are spacelike separated, as required by the Bose statistics of the spin-two field. As discussed in section~\ref{sec:phasespacereform}, the use of time-compact test tensors in (iv) is necessary to allow for certain effects of nontrivial topology. 
An analogous modification to the axioms for the electromagnetic field is required if the 
topological restrictions imposed in~\cite{DimockEM,CJFMJP} are relaxed. Also note that (ii) corresponds to item (ii) of Theorem~\ref{thm:relations} for real-valued $\bgamma$. 

The algebra of observables $\mathcal{A} (M,\bg)$, where $(M,\bg)$ is the background spacetime, will consist of finite linear combinations of finite products of $[\bgamma](\bfab)$, $[\bgamma](\bfab)^{*}$ and a unit $\II$ obeying the above relations. More formally, it may be constructed as follows: we first  form the free unital $*$-algebra generated by symbols $[\bgamma](\bfab)$ together with a unit $\II$. However, this algebra is too large as it does not take into account the relations. To impose them we quotient by the two-sided $*$-ideal they generate.\footnote{We use the same notation for $[\bgamma](\bfab)$ and the corresponding equivalence class in the quotient.}

The algebra $\mathcal{A}(M,\bg)$ may also be constructed from a different perspective, as the 
application of a quantization functor to the classical complexified phase space $\mathscr{P}(M)$ with symplectic product $\omega$ (see~\cite[Sec.~5]{FewVer:dynloc2} for the definition and properties of this functor). This is particularly convenient in the discussion of covariance below. The equivalence of these approaches follows directly from Proposition~\ref{prop:phase_reform}.

We now state and prove a time-slice property. This shows that the algebra is generated
by smeared fields with smearings supported in any slice around a Cauchy surface.
\begin{thm}\label{thm:timeslice}
Given any connected causally convex open neighbourhood $\mathscr{N}$ of a spacelike Cauchy surface $\Sigma$ and a $\bfab \in C^{\infty}_{0}(S^{0}_{2}(M))$ with $\nabla\cdot\bfab=0$ then there exists a $\tilde{\bfab} \in C^{\infty}_{0}(S^{0}_{2}(\mathscr{N}))$ with $\nabla\cdot\tilde{\bfab}=0$ and a $\bh \in C^{\infty}_{0}(S^{0}_{2}(M))$ such that 
\begin{equation}
\bfab=\tilde{\bfab}-2 L(\bh).
\end{equation}
This entails that $[\bgamma](\tilde{\bfab})=[\bgamma](\bfab)$ in $\mathcal{A}(M,\bg)$.
\end{thm}
\begin{proof} This follows the method used in the electromagnetic case from \cite[Prop.~A.3(b)]{CJFMJP}. By assumption, $\mathscr{N}$
is globally hyperbolic, and we may choose Cauchy surfaces $\Sigma^\pm$ for $M$  with $\Sigma^\pm\subset I^\pm(\Sigma)\cap\mathscr{N}$.
Take two scalar functions $\chi^{\pm} \in C^{\infty}(M)$ satisfying $\chi^{+}=1$ in $J^{+}({\Sigma^+})$, $\chi^{+}=0$ in $J^{-}({\Sigma^-})$ and $\chi^{+}+\chi^{-}=1$. Define
\begin{equation}\label{eq:LcEf}
\tilde{\bfab}:=2L(\chi^{+}\bE \overline{\bfab}),
\end{equation}
which satisfies $\nabla\cdot\tilde{\bfab}=0$ by the linearized Bianchi identity and is compactly supported within $\mathscr{N}$ ($\chi^{+}\bE \overline{\bfab}$ evidently vanishes to the past of $\mathscr{N}$ and coincides with a
de Donder solution to the linearized equations to the future of  $\mathscr{N}$ by hypothesis
on $\bfab$). Note that \eqref{eq:LcEf} implies that  $2L(\chi^{-}\bE \overline{\bfab})=-\tilde{\bfab}$. 
By Lemma~\ref{lem:Efpastfuture} we have $-\bE^{+}\overline{\tilde{\bfab}} \thicksim \chi^{+}
\bE \overline{\bfab}$, $\bE^{-}\overline{\tilde{\bfab}} \thicksim \chi^{-}\bE \overline{\bfab}$ and hence \begin{equation}
\label{eqn:EfgaugeequivEf}
\bE \overline{\tilde{\bfab}}=\bE \overline{\bfab} +\pounds_{\bw}\bg
\end{equation}
for some $\pounds_{\bw}\bg\in \mathscr{G}(M)$. As $\nabla\cdot\tilde{\bfab}=\nabla\cdot\bfab=0$, both $\bE \overline{\tilde{\bfab}}$ and $\bE \overline{\bfab}$ are de Donder solutions and so $\bw$ solves $(\Box+\Lambda )\bw=0$
(see the remarks following Theorem~\ref{thm:deDsplit}); hence by \cite[Thm~3.4.7]{Bar}, $\bw=\tilde{\bE}\bv$ for some $\bv \in C^{\infty}_{0}(T^{1}_{0}(M))$. Substituting this result into \eqref{eqn:EfgaugeequivEf} and using Lemma~\ref{lem:liederivE} gives $\bE(\overline{\tilde{\bfab}}-\overline{\bfab}-\pounds_{\bv}\bg)=0$. Using Lemma~\ref{lem:fequalsLh} gives the desired result.~\end{proof}

Finally, we note that the theory could also be quantized by means of the Weyl algebra. This
results in a unital $C^*$-algebra $\mathcal{W}(M,\bg)$ generated by elements $W([\bgamma])$
($[\bgamma]\in\mathscr{P}_{\mathbb{R}}(M)$) obeying the relations
\[
W(0)= \II, \quad W([\bgamma])^*=W(-[\bgamma]), \quad 
W([\bgamma]+[\bgamma']) = e^{i\omega([\bgamma],[\bgamma'])/2} W([\bgamma])W([\bgamma']);
\]
and with the unique norm making the resulting system a $C^*$-algebra -- for an explicit
construction of the algebra see, for example,~\cite{ManVer1968} or~\cite{Bar} (example 4.2.2 of that reference)
in the case of non-degenerate $\omega$, or~\cite{BiHoRi2004} in the degenerate case.

\subsection{Covariance}

We briefly discuss the extent to which the theory we have constructed may be formulated
as a locally covariant theory in the functorial sense introduced by Brunetti, Fredenhagen and Verch (BFV)~\cite{BFV}; however, for brevity, we will not emphasize the categorical structures here (they are 
easily reinserted). Let $(M,\bg_M)$ and $(N,\bg_N)$ be globally hyperbolic spacetimes solving the vacuum Einstein equation with cosmological constant $\Lambda$, and assume that these spacetimes are endowed with time-orientations and orientations, which we leave implicit in what follows. Consider
any smooth embedding $\psi:M\to N$ that is an isometry, preserves the (time-)orientation  and has a causally convex image. We restrict to 
spacetimes for which $\mathscr{P}(M)$ and $\mathscr{P}(N)$ are both
weakly non-degenerate, though we do not assume that $M$ and $N$ have
compact Cauchy surfaces.\footnote{Note, however, that if $M$ has compact Cauchy surface, the
existence of the embedding $\psi$ entails that $N$ also has compact Cauchy surfaces, which
are oriented-diffeomorphic to those of $M$; see~\cite[Prop.~2.3(a)]{FewVer:dynloc_theory}.}
 
Owing to the isomorphism in Proposition~\ref{prop:phase_reform}, the push-forward 
$\psi_*$ of compactly supported tensor fields induces a linear map $\mathscr{P}(\psi):
\mathscr{P}(M)\to\mathscr{P}(N)$ so that $\mathscr{P}(\psi)[\bE_M\overline{\bfab}] = [\bE_N\psi_*\overline{\bfab}]$,
provided that $\psi_*\hat{\mathscr{L}}(M)\subset \hat{\mathscr{L}}(N)$. This raises the following question:
given any $\bk\in C^\infty_{TC}(S^0_2(M))$ so that $L_M(\bk)$ is compactly supported, 
does there exist $\bk'\in C^\infty_{TC}(S^0_2(N))$ so that $L_N(\bk')=\psi_*L_M(\bk)$?
Clearly, if $\bk$ is compactly supported,  $\bk'=\psi_*\bk$ has the required property, so 
$\psi_*\mathscr{L}(M)\subset\mathscr{L}(N)\subset \hat{\mathscr{L}}(N)$.
However, this argument cannot be used if $\bk$ is a general element of $C^\infty_{TC}(S^0_2(M))$,
and it seems that there may be a genuine obstruction in some (though not all) cases where $\mathscr{L}(M)$ is a  proper subspace of $\hat{\mathscr{L}}(M)$. This is analogous to the distinction between
de Rahm cohomology with or without compact support; an equivalent question here is
whether or not $\psi^*\hat{\mathscr{G}}(N)$ exhausts $\hat{\mathscr{G}}(M)$. 

If, indeed, $\psi_*\hat{\mathscr{L}}(M)\subset \hat{\mathscr{L}}(N)$ then the map
$\mathscr{P}(\psi)$ is easily seen to be symplectic, as
\begin{align}
\omega_N(\mathscr{P}(\psi)[\bE_M\overline{\bfab}],\mathscr{P}(\psi)[\bE_M\overline{\bfab'}]) &= 
\omega_N([\bE_N\psi_*\overline{\bfab}], [\bE_N\psi_*\overline{\bfab'}]) \notag\\ & = 
-\frac{1}{2}\bE_N(\psi_*\bfab,\psi_*\overline{\bfab'})
 = -\frac{1}{2}
\bE_M(\bfab,\overline{\bfab'})  
\notag\\ &= \omega_M([\bE_M\overline{\bfab}],[\bE_M\overline{\bfab'}])
\end{align}
for arbitrary $\bfab,\bfab'\in\mathscr{F}(M)$. It is therefore also injective because
$\mathscr{P}(M)$ was assumed to be weakly non-degenerate. It is also clear that $\mathscr{P}(\psi)$
commutes with complex conjugation. Consequently, the general properties of the quantization
functor used to construct $\mathcal{A}(M,g_M)$ from $\mathscr{P}(M)$ (see~\cite{FewVer:dynloc2}) 
entail the existence of an injective, unit-preserving $*$-homomorphism $\mathcal{A}(\psi):\mathcal{A} (M,\bg_M)\to \mathcal{A} (N,\bg_N)$, uniquely
determined by the property $\mathcal{A}(\psi)[\bgamma]_M(\bfab)= [\bgamma]_N(\psi_*\bfab)$. 
If we consider an additional embedding $\varphi:M'\to M$ of the above type, it is
evident that $\mathscr{A}(\psi)\circ\mathscr{A}(\varphi) = \mathscr{A}(\psi\circ\varphi)$;
it is also clear that $\mathscr{A}$ maps the identity mapping of $M$ to the identity mapping
of $\mathscr{A}(M,\bg_M)$. Moreover, if the image of $\psi$ contains a Cauchy surface of $N$, 
then the map $\mathscr{A}(\psi)$ is surjective by Theorem~\ref{thm:timeslice} and hence
an isomorphism. 

This discussion may be summarised in the language of~\cite{BFV} by saying that  
$\mathscr{A}$ defines a covariant functor from a certain category of globally hyperbolic spacetimes
to the category of unital $*$-algebras with unit-preserving injective $*$-homomorphisms as
morphisms, and that this theory has the time-slice property. Similarly, the quantization of $\mathscr{P}_{\mathbb{R}}(M)$ in terms of Weyl algebras 
defines a functor to the
category of unital $C^*$-algebras. However, the category of spacetimes is a subcategory
of that usually studied  in the BFV formalism: the spacetimes themselves are restricted
to those cosmological vacuum solutions on which the phase space of linearized gravity is
weakly non-degenerate, while the morphisms must be restricted to permit
only those embeddings $\psi:M\to N$ for which $\psi_*\hat{\mathscr{L}}(M)\subset \hat{\mathscr{L}}(N)$. A more geometrical characterization of this class would be desirable. 

The restriction to cosmological vacuum solutions has an important consequence: it is
not possible to formulate the {\em relative Cauchy evolution} that plays an important
part in the BFV formalism -- see~\cite{BFV} 
and~\cite{FewVer:dynloc_theory,FewVer:dynloc2} -- but which requires the 
freedom to consider arbitrary (sufficiently small) compactly supported metric perturbations. 
This is related to the well-known absence of a local stress-energy tensor for the
gravitational field. Finally, the detailed study of circumstances under which 
$\mathscr{P}(M)$ is weakly non-degenerate, and of those morphisms 
$\psi:M\to N$ for which $\mathscr{P}(\psi)$ is well defined, are evidently interesting problems with a cohomological flavour and deserve further investigation.

{\noindent{\em Acknowledgements:} The authors thank Atsushi Higuchi for many extremely useful discussions on this subject and Vincent Moncrief for invaluable assistance (via Atsushi) in understanding non-degeneracy of the symplectic product. We also thank Ian McIntosh for several discussions on geometrical results and Benjamin Lang for discussions concerning analogues with electromagnetism. We are also grateful for comments received from Stanley Deser, Andrew Waldron and Ko Sanders.}


\appendix
\section{Transverse-traceless gauge}\label{appx:TT}

\renewcommand\thethm{\Alph{section}.\arabic{thm}}

\subsection{Conventions}
\label{sec:forms}

As the proof of Theorem~\ref{thm:TTsplit} makes use of differential forms we briefly
summarize the conventions employed. We work with the smooth spacetime $(M,\bg)$ satisfying all of the topological criteria from section~\ref{sec:lingrav}. Our conventions follow~\cite{AMR} and are consistent with those of \cite{CJFMJP,MJP}: for the space of $p$-forms on $M$ we use the notation $\Omega^{p}(M)$, the space of compactly supported $p$-forms on $M$ is denoted $\Omega^{p}_{0}(M)$ and the space of $p$-forms that are spacelike-compact is denoted by $\Omega^{p}_{SC}(M)$.

The wedge-product of $\balpha \in \Omega^{p}(M)$ and $\bbeta \in \Omega^{q}(M)$ is $\balpha \wedge \bbeta \in \Omega^{p+q}(M)$, which is given by
\begin{equation}
(\alpha \wedge \beta)_{a_{1} \dots a_{p+q}} = \frac{(p+q)!}{p!q!}\alpha_{[a_{1}\dots a_{p}}\beta_{a_{p+1}\dots a_{p+q}]}.
\end{equation}

The Hodge star operator is the map $*: \Omega^{p}(M) \to \Omega^{n-p}(M)$ uniquely defined~\cite[Prop.~6.2.12]{AMR} by the condition
$\balpha \wedge * \bbeta = (\balpha,\bbeta)_{\bg}dvol_{\bg}$
for all $\balpha,\bbeta \in \Omega^{p}(M)$. Its square is, by~\cite[Prop.~6.2.13]{AMR}, equal to
\begin{equation}
(*)^{2}=(-1)^{p(n-p)+s},
\end{equation}
where $n$ is the dimension of the manifold and $s$ is the index of the metric $\bg$. On a four-dimensional Lorentzian manifold (in $-+++$ signature) this becomes $(*)^{2} = (-1)^{p+1}$, while on a spacelike Cauchy surface thereof it reduces to $(*)^{2} =1$.

There is a standard pairing, see \cite[p.~538]{AMR}, between $p$-forms on a manifold $M$: given $\bw \in \Omega^{p}(M)$ and $\bfab \in \Omega^{p}_{0}(M)$ one defines
\begin{equation}
\label{eqn:formspairing}
\langle \bw , \bfab \rangle_{M} := \int_{M}{\bw \wedge * \bfab}.
\end{equation}
If $M$ is compact then the restriction that $\bfab$ be compactly supported can be dropped.

The exterior derivative $\text{d}:\Omega^{p}(M) \to \Omega^{p+1}(M)$ is given by
\begin{equation}
(\text{d} \alpha)_{a_{1} \dots a_{p+1}}=(p+1)\nabla_{[a_{1}}\alpha_{a_{2}\dots a_{p+1}]},
\end{equation}
where $\balpha \in \Omega^{p}(M)$. The codifferential $\delta :\Omega^{p}(M)\to\Omega^{p-1}(M)$ is, by \cite[Dfn~6.5.21]{AMR}, given by
\begin{equation}
\label{eqn:codifferential}
\delta = (-1)^{n(p-1)+s+1}*\text{d}*
\end{equation}
and always annihilates $0$-forms. On the spacetime $(M,\bg)$, $\delta = * \text{d} *$, while on spacelike Cauchy surfaces, $\delta = (-1)^{p} * \text{d} *$.

By using Stokes' Theorem one can show that
\begin{equation}
\label{eqn:dtodelta}
\langle \text{d} \balpha,\bbeta \rangle=\langle \balpha, \delta \bbeta \rangle
\end{equation}
for $\balpha \in \Omega^{p-1}(M)$ and $\bbeta \in \Omega^{p}(M)$, provided at least one of them is
compactly supported.

\subsection{The transverse-traceless gauge}

By using differential forms we can make use of the methodology of \cite[Prop.~2.6]{MJP}, which deals with obtaining the Lorenz gauge in electromagnetism, where there is an identical hyperbolic equation and a constraint similar to \eqref{eqn:remDDw_and_traceconstraintw}.

As in \cite[Sec.~2.4]{MJP}, define $i: \Sigma \to M$ to be the embedding of the Cauchy surface $\Sigma$ in the spacetime $M$ and define the following forms on a Cauchy surface $\Sigma$:
\begin{align}
\label{eqn:initial}
{\boldsymbol{w_{(0)}}} &:= i^{*} \bw \in \Omega^{1}_{0}(\Sigma) \\
\label{eqn:initiald}
{\boldsymbol{ w_{(d)}}} &:= -*i^{*}* \text{d} \bw \in \Omega^{1}_{0}(\Sigma)\\
\label{eqn:initialdelta}
w_{(\delta)} &:= i^{*} \delta \bw \in \Omega^{0}_{0}(\Sigma) \\
\label{eqn:initialn}
w_{(n)} &:= -*i^{*}* \bw \in \Omega^{0}_{0}(\Sigma),
\end{align}
which together constitute the Cauchy data on $\Sigma$ for $\bw$, meaning they correspond to $\bw|_{\Sigma}$ and $\nabla_{\bn}\bw|_{\Sigma}$. We sometimes use the notation $\rho_{(0)}$, $\rho_{(d)}$, $\rho_{(\delta)}$ and $\rho_{(n)}$ from \cite[Sec.~2.4]{MJP} for the linear maps \eqref{eqn:initial}, \eqref{eqn:initiald}, \eqref{eqn:initialdelta} and \eqref{eqn:initialn} respectively, applied to one-forms and zero-forms. The forms on a Cauchy surface corresponding to the Cauchy data for a zero-form $\gamma$ are 
\begin{align}
\label{eqn:gamma0}
\gamma_{(0)} &:= \rho_{(0)}\gamma \in \Omega^{0}_{0}(\Sigma) \\
\label{eqn:gammad}
\gamma_{(\text{d})} &:= \rho_{(\text{d})}\gamma \in \Omega^{0}_{0}(\Sigma).
\end{align}
(Note: $\gamma_{(\delta)}$ and $\gamma_{(n)}$ are automatically zero.) Writing \eqref{eqn:remDDw_and_traceconstraintw} in forms notation gives
\begin{align}
\label{eqn:remDDforms}
-(\delta \text{d}+\text{d} \delta)\bw+2\Lambda \bw &= 0, \\
\label{eqn:TTconstraintforms}
\delta \bw &=\frac{1}{2}\gamma
\end{align}
respectively. In \eqref{eqn:TTconstraintforms} we used $\delta \bw=-\nabla^{a}w_{a}$ and in \eqref{eqn:remDDforms} the `extra' $\Lambda \bw$ comes from non-commutativity of covariant derivatives and using \eqref{eqn:Riccitenscalvalue}.

We will require scalar and one-form Green's identities, see \cite[Appx A]{Furlani:1999}, \cite[Sec.~2.4]{MJP},\footnote{The formulae in~\cite{Furlani:1999,MJP} are modified by changing to our $-+++$
signature convention. Note that there are sign errors in Eq.~(2.21) in \cite{MJP},
and that~\cite{Furlani:1999} uses the retarded-minus-advanced propagator, whereas we use
the advanced-minus-retarded.} which link solutions of hyperbolic equations to their initial data on a spacelike Cauchy surface $\Sigma$. For $\bw \in \Omega^{1}_{SC}(M)$ solving \eqref{eqn:remDDforms} we have
\begin{equation}
\label{eqn:greensidentity}
\langle \bw,\bfab \rangle_{M}=\langle {\boldsymbol{ w_{(0)}}},\rho_{(d)}\bE\bfab \rangle_{\Sigma} + \langle w_{(\delta)},\rho_{(n)}\bE\bfab \rangle_{\Sigma}  - \langle {\boldsymbol{ w_{(d)}}},\rho_{(0)}\bE\bfab \rangle_{\Sigma} - \langle w_{(n)},\rho_{(\delta)}\bE\bfab \rangle_{\Sigma},
\end{equation}
where $\bfab \in \Omega^{1}_{0}(M)$ and $\bE$ is the advanced-minus-retarded solution operator for the differential operator $-(\delta \text{d} + \text{d} \delta) + 2\Lambda$ acting on $1$-forms (see section~\ref{sec:greensoperators} for further details on Green's operators). 
The scalar case is \eqref{eqn:scalargreensidentity}.

We now show what constraints the Cauchy data need to satisfy in order that the solution to \eqref{eqn:remDDforms} also satisfies \eqref{eqn:TTconstraintforms}; we adapt \cite[Prop.~2.6]{MJP} to achieve this.
\begin{thm}
\label{thm:TTgauge}
Suppose $\bw \in \Omega^{1}_{SC}(M)$ solves $(-(\delta \text{d} + \text{d} \delta) + 2\Lambda) \bw=0$ and $\bgamma$ is a de Donder solution, then $\delta \bw=\frac{1}{2} \gamma$ if and only if $w_{(\delta)}=\frac{1}{2}\gamma_{(0)}$ and $\delta {\boldsymbol{ w_{(d)}}}+2\Lambda w_{(n)}=\frac{1}{2} \gamma_{(\text{d})}$.
\end{thm}
\begin{proof} $(\Rightarrow)$ The pull-back of the  
constraint $\delta \bw=\frac{1}{2} \gamma$ 
to the Cauchy surface gives $w_{(\delta)}=\frac{1}{2}\gamma_{(0)}$,
while applying $-*i^{*}*\text{d}$ to the constraint gives  
\begin{equation}
\label{eqn:rhodconstraintexplic}
-*i^{*}*\text{d} \delta \bw =  
\rho_{(d)} \delta \bw = \frac{1}{2} \rho_{(d)} \gamma = 
\frac{1}{2} \gamma_{(\text{d})}.
\end{equation}
Using \eqref{eqn:remDDforms}, this yields 
\begin{equation}
\label{eqn:rhodKGw}
*i^{*}* \delta \text{d} \bw - 2\Lambda * i^{*}* \bw = \frac{1}{2} \gamma_{(\text{d})},
\end{equation}
which may be rewritten as $\delta {\boldsymbol{ w_{(d)}}} + 2 \Lambda w_{(n)} = \frac{1}{2} \gamma_{(\text{d})}$ by using the definitions of the various quantities and the fact that $d$ commutes
with $i^{*}$.

$(\Leftarrow)$ To prove that such a $\bw$ will satisfy the constraint $\delta \bw = \frac{1}{2}\gamma$ globally on $(M,\bg)$ we begin by taking an arbitrary $f \in \Omega^{0}_{0}(M)$ and computing
\begin{multline}
\label{eqn:Greenexpwdf}
\langle \delta \bw,f \rangle_{M} = \langle \bw,\text{d}f \rangle_{M}=\langle {\boldsymbol{ w_{(0)}}},\rho_{(d)}\bE\text{d}f \rangle_{\Sigma} + \langle w_{(\delta)},\rho_{(n)}\bE\text{d}f \rangle_{\Sigma} \\ - \langle {\boldsymbol{ w_{(d)}}},\rho_{(0)}\bE\text{d}f \rangle_{\Sigma} - \langle w_{(n)},\rho_{(\delta)}\bE\text{d}f \rangle_{\Sigma},
\end{multline}
where we use \eqref{eqn:dtodelta} and \eqref{eqn:greensidentity}.

From \cite[Prop.~2.1]{MJP} we know that $\bE\text{d}= \text{d}E$, where $E$ is the 
advanced-minus-retarded solution operator for the differential operator $-(\delta \text{d} + \text{d} \delta) + 2\Lambda =-\delta \text{d}+ 2\Lambda$ acting on $0$-forms. Using this and $\rho_{(d)}\text{d}=0$, \eqref{eqn:Greenexpwdf} reduces to
\begin{equation}
\langle \delta \bw,f \rangle_{M}=\langle w_{(\delta)},\rho_{(n)}\text{d}Ef \rangle_{\Sigma} - \langle {\boldsymbol{ w_{(d)}}},\rho_{(0)}\text{d}Ef \rangle_{\Sigma} - \langle w_{(n)},\rho_{(\delta)}\text{d}Ef \rangle_{\Sigma}.
\end{equation}
For the second term on the right-hand side we can use that the pullback and the exterior derivative commute, and then \eqref{eqn:dtodelta} to obtain
\begin{align}
\langle \delta \bw,f \rangle_{M} &=\langle w_{(\delta)},\rho_{(n)}\text{d}Ef \rangle_{\Sigma} - \langle \delta {\boldsymbol{ w_{(d)}}},\rho_{(0)} Ef \rangle_{\Sigma} - \langle w_{(n)},\rho_{(\delta)}\text{d}Ef \rangle_{\Sigma}\\
\label{eqn:deltawfcauchy}
&= \frac{1}{2}\langle \gamma_{(0)},\rho_{(n)}\text{d}Ef \rangle_{\Sigma} - \frac{1}{2} \langle \gamma_{(d)} ,\rho_{(0)} Ef \rangle_{\Sigma}\notag \\ &\qquad\qquad+ 2\Lambda \langle w_{(n)},\rho_{(0)} Ef  \rangle - \langle w_{(n)},\rho_{(\delta)}\text{d}Ef \rangle_{\Sigma},
\end{align}
where in \eqref{eqn:deltawfcauchy} we substituted the restrictions on the Cauchy data. The last two terms cancel because
$\rho_{(\delta)}\text{d}Ef = \rho_{(0)}\delta \text{d}Ef = 2\Lambda\rho_{(0)}Ef$.
Therefore
\begin{equation}\label{eqn:deltawf}
\langle \delta \bw,f \rangle_{M} =\frac{1}{2}\langle \gamma_{(0)},\rho_{(d)}Ef \rangle_{\Sigma} - \frac{1}{2} \langle \gamma_{(d)} ,\rho_{(0)} Ef \rangle_{\Sigma},
\end{equation}
where we have used that, in this case, $\rho_{(d)}=\rho_{(n)}\text{d}$.

The trace of a de Donder solution $\bgamma$ satisfies the scalar wave equation \eqref{eqn:traceeqnmotion}, 
which in forms notation is $-\delta \text{d}\gamma+2\Lambda \gamma=0$. 
Therefore we can use the scalar version of \eqref{eqn:greensidentity} [recall that $\gamma_{(\delta)}=\gamma_{(n)}=0$], i.e., 
\begin{equation}
\label{eqn:scalargreensidentity}
\langle \gamma, f \rangle_{M} =\langle \gamma_{(0)},\rho_{(d)}Ef \rangle_{\Sigma} -\langle \gamma_{(d)},\rho_{(0)}Ef \rangle_{\Sigma}
\end{equation}
for any $f\in \Omega_0^0(M)$. Comparing this with \eqref{eqn:deltawf} gives 
$\langle \delta \bw,f \rangle_{M} = \frac{1}{2}\langle \gamma,f \rangle_{M}$ 
for all $f \in \Omega^{0}_{0}(M)$, 
and hence $\delta \bw=\frac{1}{2} \gamma$. \end{proof}

We may now give the proof of the splitting $\mathscr{S}(M) = \mathscr{S}^{TT}(M) + \mathscr{G}(M)$
for $\Lambda\neq 0$.

\begin{proofof}{Theorem~\ref{thm:TTsplit}} We know from Corollary~\ref{cor:deDsplitS} that $\mathscr{S}(M)=\mathscr{S}^{dD}(M)+\mathscr{G}(M)$, therefore if we can decompose the space of de Donder solutions as $\mathscr{S}^{dD}(M)=\mathscr{S}^{TT}(M)+\mathscr{G}(M)\cap \mathscr{S}^{dD}(M)$ then we can achieve \eqref{eqn:TTsplit}. Given a perturbation $\bgamma \in \mathscr{S}^{dD}(M)$ on a cosmological vacuum spacetime $(M,\bg)$ with $\Lambda\neq 0$, the constraints of Theorem~\ref{thm:TTgauge} are satisfied by $w_{(0)}=\frac{1}{4\Lambda}\text{d} \gamma_{(0)}$, $w_{(\text{d})}=0$, $w_{(n)} = \frac{1}{4\Lambda} \gamma_{(\text{d})}$ and $w_{(\delta)}=\frac{1}{2}\gamma_{(0)}$ as Cauchy data. The solution with this data is $\bw = \frac{1}{4\Lambda} \text{d} \gamma$, which corresponds to the choice in equation~(9) of~\cite{AHSKlargedist} for de Sitter spacetime. Therefore appropriate Cauchy data exists and one may gauge transform from the de Donder gauge to the transverse-traceless gauge. 
\end{proofof}

In the case $\Lambda=0$, the  second constraint of Theorem~\ref{thm:TTgauge} reduces to $\delta {\boldsymbol{ w_{(d)}}}=\frac{1}{2} \gamma_{(\text{d})}$, which becomes a cohomological problem. The scalar $\gamma_{(\text{d})}$ is co-closed, $\delta \gamma_{(\text{d})}=0$, but is it co-exact?
Equivalently, we must solve 
\begin{equation}
\label{eqn:exact*gammad}
\text{d}(* {\boldsymbol{ w_{(d)}}}) = -\frac{1}{2} * \gamma_{(\text{d})},
\end{equation}
in which $*\gamma_{(\text{d})}$ is a $3$-form on $\Sigma$ and necessarily closed.

There are two cases to consider: depending on whether or not $\bw$ has compact support on Cauchy surfaces. If as we assume, $\bw$ has compact support on Cauchy surfaces then from \cite[Thm~7.5.19(i)]{AMR}, $* \gamma_{(\text{d})}$ is exact if and only if
\begin{equation}
\int_{\Sigma}{* \gamma_{(\text{d})}}=0.
\end{equation}
If, on the other hand, $\Sigma$ is non-compact and $\bw$ is allowed to have non-compact support on $\Sigma$,
\cite[Thm~7.5.19(iii)]{AMR} gives that $H^{3}(\Sigma)=0$ and so $* \gamma_{(\text{d})}$ is exact and the TT gauge may be attained as is standard, e.g.,  in Minkowski spacetime.

\section{Non-degeneracy}
\label{sec:ADMnondegen}

\subsection{Background on the ADM formalism}

In order to prove Theorem~\ref{thm:radical} we need to appeal to the results of \cite{Moncriefdecomp,FM}, which use the ADM formalism (for original references, see
 \cite{ADM}). This formalism puts general relativity into the form of a dynamical system, where given a three-dimensional  smooth manifold $\Sigma$ and Cauchy data $(\bh,\bvarpi)\in C^{\infty}(S^{0}_{2}(\Sigma))\times C^{\infty}(\tilde{S}^{2}_{0}(\Sigma))$ we can obtain a solution $((-\epsilon,\epsilon)\times \Sigma,\bg)$ to Einstein's equation. Here 
$\tilde{S}^{2}_{0}(\Sigma)$ is the space of smooth second rank contravariant tensor densities on $\Sigma$, $\bh$ is the spatial metric on $\Sigma$ and $\varpi^{ab}=\sqrt{h}\left(k^{ab}-\frac{1}{2}h^{ab}k \right)$ is the conjugate momentum, where $\bk$ is the desired extrinsic curvature\footnote{Our convention differs from~\cite{FM}; however, the overall definitions of $\bvarpi$ coincide.} of $\Sigma$ in the solution spacetime and $h$ is the determinant of the metric $\bh$. In fact, to obtain a solution, one also needs to specify a lapse function and a shift vector field on $\Sigma$, which can both be time-dependent and are freely specifiable and non-dynamical. Together they make up the components of a vector field whose integral curves provide a flow of time in spacetime. The spacetime metric $\bg$ is constructed from the lapse and shift as well as the evolved spatial
metric $\bh$ obtained from solving the ADM equations \eqref{eqn:ADMevol} below.

The initial data $(\bh,\bvarpi)$ are not freely specifiable as they need to satisfy constraints given by the map $\Phi:C^{\infty}( S^{0}_{2}(\Sigma)) \times C^{\infty}(\tilde{S}^{2}_{0}(\Sigma)) \to C^{\infty}(\Sigma) \times C^{\infty}(T^{1}_{0}(\Sigma))$, where
\begin{equation}
\label{eqn:GRconstraints}
\Phi(\bh,\bvarpi)=(\mathscr{H}(\bh,\bvarpi),\bdelta(\bh,\bvarpi))
\end{equation}
and the Hamiltonian and momentum constraints are
\begin{equation}
\label{eqn:IVCH}
\mathscr{H}(\bh,\bvarpi)= -R^{(3)}(\bh)+\frac{\varpi^{ab}\varpi_{ab}}{h}-\frac{\varpi^{2}}{2h}+2\Lambda
\end{equation}
and
\begin{equation}
\label{eqn:IVCM}
\delta^{a}(\bh,\bvarpi)= D_{b}\left(\frac{\varpi^{ab}}{\sqrt{h}} \right)
\end{equation}
respectively. Here $R^{(3)}(\bh)$ is the Ricci scalar for the metric $\bh$ and $D_{a}$ is the covariant derivative associated with $\bh$. Vanishing of \eqref{eqn:GRconstraints} is a necessary condition for a spacetime to be a solution to the vacuum Einstein equation with cosmological constant.

For linearized gravity in the ADM formalism, one considers a one-parameter family of 
Cauchy data $(\bh(\lambda),\bvarpi(\lambda))$ and takes the derivative at $\lambda=0$. Thus the Cauchy data for the linearized ADM equations are $(\bgamma^{(3)},\bp)=\left(\frac{\partial \bh(\lambda)}{\partial \lambda},\frac{\partial \bvarpi(\lambda)}{\partial \lambda} \right)_{\lambda=0}$ and the linearized constraints are the components of the derivative of the constraint map \eqref{eqn:GRconstraints} at $(\bh,\bvarpi)$. One should note that to solve the linearized ADM equations one needs to specify a linearized lapse function and linearized shift vector field; as in the full non-linear case they are non-dynamical and freely specifiable.

From now on we assume that the background is a solution to the vacuum Einstein equation with cosmological constant and so $\Phi(\bh,\bvarpi)=0$. The linearized constraints are the derivative of the constraint map \eqref{eqn:GRconstraints},
\begin{equation}
D\Phi(\bh,\bvarpi): C^{\infty}(S^{0}_{2}(\Sigma)) \times C^{\infty}(\tilde{S}^{2}_{0}(\Sigma)) \to C^{\infty}(\Sigma) \times C^{\infty}(T^{1}_{0}(\Sigma))
\end{equation}
 evaluated at $(\bh,\bvarpi)$, where
\begin{equation}
\label{eqn:lincon}
D\Phi(\bh,\bvarpi)\begin{pmatrix}\bgamma^{(3)}\\ \bp\end{pmatrix}= \begin{pmatrix}D\mathscr{H}(\bh,\bvarpi)\begin{pmatrix}\bgamma^{(3)}\\\bp\end{pmatrix}\\ D \bdelta(\bh,\bvarpi)\begin{pmatrix}\bgamma^{(3)}\\ \bp\end{pmatrix}\end{pmatrix}
\end{equation}
whose actions are given by
\begin{multline}
\label{eqn:linconH}
D\mathscr{H}(\bh,\bvarpi)\begin{pmatrix}\bgamma^{(3)}\\ \bp\end{pmatrix}=\frac{1}{h}\left[-\left(\varpi^{ab}\varpi_{ab}-\frac{1}{2}\varpi^{2}\right)\gamma^{(3)}+2\left(\varpi_{ab}p^{ab}-\frac{1}{2}\varpi p\right) \right. \\ \left. +2\left(\varpi^{ac}\varpi_{cb}-\frac{1}{2}\varpi \varpi^{ab}\right)\gamma^{(3)}_{ab} \right]
-\left(D^{a}D^{b}\gamma^{(3)}_{ab}-D^{a}D_{a}\gamma^{(3)} - R^{(3)ab}\gamma^{(3)}_{ab}\right) 
\end{multline}
and
\begin{equation}
D \bdelta(\bh,\bvarpi)\begin{pmatrix}\bgamma^{(3)}\\ \bp\end{pmatrix} =\frac{1}{\sqrt{h}}\left[ 2D_{b}p^{ab}+\varpi^{bc}\left(D_{c}\gamma^{(3)a}_{\hspace{0.1cm}b}+D_{b}\gamma^{(3)a}_{\hspace{0.1cm}c} -D^{a}\gamma^{(3)}_{bc}\right)\right],
\end{equation}
where $\gamma^{(3)}=h^{ab}\gamma^{(3)}_{ab}$, $\varpi=h_{ab}\varpi^{ab}$ and $p=h_{ab}p^{ab}$.
To get these into a form analogous to that in \cite{Moncriefdecomp,FM} we evaluate \eqref{eqn:linconH}  on the constraint surface to give
\begin{multline}
\label{eqn:linconFM}
D\mathscr{H}(\bh,\bvarpi)\begin{pmatrix}\bgamma^{(3)}\\ \bp\end{pmatrix}=\frac{1}{h}\left[-\frac{1}{2}\left(\varpi^{ab}\varpi_{ab}-\frac{1}{2}\varpi^{2}\right)\gamma^{(3)}+2\left(\varpi_{ab}p^{ab}-\frac{1}{2}\varpi p\right) \right. \\ \left. +2\left(\varpi^{ac}\varpi_{c}^{\hspace{0.1cm}b}-\frac{1}{2}\varpi \varpi^{ab}\right)\gamma^{(3)}_{ab} \right]
-\left[D^{a}D^{b}\gamma^{(3)}_{ab}-D^{a}D_{a}\gamma^{(3)} \right. \\ - \left.\left(R^{(3)ab}-\frac{1}{2}h^{ab}R^{(3)}+\Lambda h^{ab}\right)\gamma^{(3)}_{ab}\right].
\end{multline}
Note that the difference between this and the $\Lambda=0$ case considered in equation~(2.8) of \cite{Moncriefdecomp} is the cosmological constant term.

We also require the following inner products, defined in \cite{Moncriefdecomp} (equations (2.4) and (2.6) in that reference). The first acts on the domain of  $D\Phi(\bh,\bvarpi)$, i.e., the vector space $C^{\infty}(S^{0}_{2}(\Sigma)) \times C^{\infty}(\tilde{S}^{2}_{0}(\Sigma))$ and is
\begin{equation}
\label{eqn:inprodhp}
\langle (\bgamma^{(3)},\bp);(\tilde{\bgamma}^{(3)},\tilde{\bp}) \rangle := \int_{\Sigma}{ \left(\sqrt{h}\gamma^{(3)}_{ab}\tilde{\gamma}^{(3)}_{cd}h^{ac}h^{bd}+\frac{1}{\sqrt{h}}h_{ac}h_{bd}p^{ab}p^{cd}\right)},
\end{equation}
where $\bgamma^{(3)},\tilde{\bgamma}^{(3)}\in C^{\infty}(S^{0}_{2}(\Sigma))$ and $\bp,\tilde{\bp}\in C^{\infty}(\tilde{S}^{2}_{0}(\Sigma))$. The second acts on the codomain
$C^{\infty}(\Sigma) \times C^{\infty}(T^{1}_{0}(\Sigma))$ of  $D\Phi(\bh,\bvarpi)$ as follows
\begin{equation}
\label{eqn:inprodCX}
\langle\! \langle (f,\bV);(\tilde{f},\tilde{\bV}) \rangle\!\rangle := \int_{\Sigma}{(f \cdot \tilde{f}+h_{ab}V^{a}V^{b})dvol_{\bh}},
\end{equation}
where $f,\tilde{f} \in C^{\infty}(\Sigma)$ and $\bV,\tilde{\bV} \in C^{\infty}(T^{1}_{0}(\Sigma))$.

The adjoint of the differential operator $D\Phi(\bh,\bvarpi)$ is calculated, using the inner products \eqref{eqn:inprodhp} and \eqref{eqn:inprodCX} and integration by parts, to be 
\begin{equation}
D\Phi(\bh,\bvarpi)^{*}\begin{pmatrix}f\\ \bV\end{pmatrix}=\begin{pmatrix}D\mathscr{H}(\bh,\bvarpi)^{*}(f)\\ D \bdelta(\bh,\bvarpi)^{*}(\bV)\end{pmatrix},
\end{equation}
where $D\mathscr{H}(\bh,\bvarpi)^{*}(f) = (\balpha,\bbeta)$ with
\begin{multline}
\alpha_{ab}=\frac{1}{h}\left[-\frac{1}{2} \left(\varpi^{cd}\varpi_{cd}-\frac{1}{2}\varpi^{2}\right)h_{ab}f \right.  \left.+2\left(\varpi_{ac}\varpi^{c}_{\hspace{0.1cm}b}-\frac{1}{2}\varpi_{ab}\varpi\right)f \right] \\
-\left[D_{a}D_{b}f-h_{ab}D^{c}D_{c}f-\left(R^{(3)}_{ab}-\frac{1}{2}h_{ab}R^{(3)}+\Lambda h_{ab}\right)f\right] 
\end{multline}
and
\begin{equation}
\beta^{ab}=2f\left(\varpi^{ab}-\frac{1}{2}\varpi h^{ab}\right).
\end{equation}
Also
\begin{equation}
D \bdelta(\bh,\bvarpi)^{*}(\bV)=\begin{pmatrix}\frac{1}{\sqrt{h}}\left(D_{c}(V^{c}\varpi_{ab})-2\varpi^{c}_{\phantom{c}(a}D_{|c|}V_{b)}\right)\\ -\sqrt{h}(D^{a}V^{b}+D^{b}V^{a})\end{pmatrix}.
\end{equation}

We now define a unitary $U:C^{\infty}(S^{0}_{2}(\Sigma))\times C^{\infty}(\tilde{S}^{2}_{0}(\Sigma)) \to C^{\infty}(S^{0}_{2}(\Sigma))\times C^{\infty}(\tilde{S}^{2}_{0}(\Sigma))$ by
\begin{equation}
U(\bgamma^{(3)},\bp):=\left(\frac{-1}{\sqrt{h}}\bp^{\flat \flat}, \sqrt{h} (\gamma^{(3)})^{\sharp \sharp} \right)
\end{equation}
so that $U \circ D\Phi(\bh,\bvarpi)^{*}$ corresponds to $\gamma(\bh,\bvarpi)\equiv\begin{pmatrix} 0 &-1 \\ 1 & 0 \end{pmatrix} \circ D\Phi(\bh,\bvarpi)^{\dagger}$ from \cite{Moncriefdecomp}, where $D\Phi(\bh,\bvarpi)^{\dagger}$ is the `new form of the adjoint' defined in equation~(4.2) of that reference. The inverse map $U^{-1}:C^{\infty}(S^{0}_{2}(\Sigma))\times C^{\infty}(\tilde{S}^{2}_{0}(\Sigma)) \to C^{\infty}(S^{0}_{2}(\Sigma))\times C^{\infty}(\tilde{S}^{2}_{0}(\Sigma))$ is given by
\begin{equation}
U^{-1}(\bgamma^{(3)},\bp)= \left( \frac{1}{\sqrt{h}}\bp^{\flat \flat},-\sqrt{h}(\bgamma^{(3)})^{\sharp \sharp} \right).
\end{equation}
The ADM evolution equations may be written as follows
\begin{equation}
\label{eqn:ADMevol}
\frac{\partial}{\partial \lambda} \begin{pmatrix} \bh(\lambda) \\ \bvarpi(\lambda)\end{pmatrix}=U^{-1} \circ D\Phi(\bh,\bvarpi)^{*}\begin{pmatrix} N \\ -\bN\end{pmatrix},
\end{equation}
where $N$ is the lapse function and $\bN$ is the shift vector field associated with the slicing. The ADM symplectic product on the background $(\bh,\bvarpi)$ is, see \cite[p.~333]{FM}, given by
\begin{equation}
\label{eqn:ADMsymprod}
\omega^{ADM}_{(\bh,\bvarpi)}((\bgamma^{(3)},\bp);(\tilde{\bgamma}^{(3)},\tilde{\bp}))=\int_{\Sigma}{(\gamma^{(3)}_{ab}\tilde{p}^{ab}-\tilde{\gamma}^{(3)}_{ab}p^{ab})d^{3}x}.
\end{equation}
Observe that
\begin{equation}
\label{eqn:ADMsympUinner}
\omega^{ADM}_{(\bh,\bvarpi)}((\bgamma^{(3)},\bp);(\tilde{\bgamma}^{(3)},\tilde{\bp}))=\langle (\bgamma^{(3)},\bp); U^{-1}(\tilde{\bgamma}^{(3)},\tilde{\bp})) \rangle.
\end{equation}

\subsection{Analogues of Moncrief's splitting theorems}

In \cite{Moncriefdecomp} it is shown that the space of initial data can be decomposed into orthogonal subspaces, using the inner product \eqref{eqn:inprodhp}. Here, we generalize these decompositions to the 
case of nonzero cosmological constant. The first splitting is
\begin{equation}
\label{eqn:firstmoncriefsplit}
C^{\infty}(S^{0}_{2}(\Sigma)) \times C^{\infty}(\tilde{S}^{2}_{0}(\Sigma))=\ker D\Phi(\bh,\bvarpi) \oplus \range D\Phi(\bh,\bvarpi)^{*},
\end{equation}
where $\ker D\Phi(\bh,\bvarpi)$ is the subspace of data satisfying the linearized constraints and $\range D\Phi(\bh,\bvarpi)^{*}$ is the unphysical data. This splitting was done for the case of $\Lambda=0$ in \cite[Sec.~3]{Moncriefdecomp} using ellipticity of the operator $D\Phi(\bh,\bvarpi)\circ D\Phi^{*}(\bh,\bvarpi)$, which is proven by showing that  $D\Phi^{*}(\bh,\bvarpi)$ has injective principal symbol and applying \cite[Thm~4.4]{BergerEbin}, which is valid on compact Riemannian manifolds. Since our modifications to the linearized constraint map and its adjoint only add a $\Lambda \bh$ term, the principal symbol will be unaffected and so the operator is still elliptic. Hence the remainder of the Moncrief argument of  \cite[Sec.~3]{Moncriefdecomp} remains valid and the first splitting holds for general $\Lambda$.

The second splitting decomposes the subspace $\ker D\Phi(\bh,\bvarpi)$ into a pure gauge subspace, meaning data for pure gauge solutions, and a physical subspace. In \cite[Sec.~IV]{MoncrieflinstabI} it is shown that data for a pure gauge solution to the linearized equations corresponding to $\pounds_\bw\bg$, on a vacuum spacetime with $\Lambda=0$, is given by
\begin{equation}
\label{eqn:puregaugedata}
(\bgamma^{(3)},\bp)_{\text{gauge}}=U \circ D\Phi(\bh,\bvarpi)^{*}\begin{pmatrix}C\\ \bX\end{pmatrix},
\end{equation}
where $C=n_{a}w^{a}$ and $X^{a}=q^{a}_{\phantom{a}b}w^{b}$ are respectively the normal (with respect to the future pointing normal vector $\bn$) and tangential projections, relative to $\Sigma$ (using the associated projection tensor $q^{a}_{\phantom{a}b}$), of the gauge vector field. The above result was initially proved via a lengthy calculation, and later by more geometrical methods, see \cite[Thm~4.7]{FM}. The result \eqref{eqn:puregaugedata} also holds on vacuum spacetimes with non-vanishing cosmological constant by following the same argument used in the proof of \cite[Thm~4.7]{FM} but instead using the vacuum ADM equations with cosmological constant \eqref{eqn:ADMevol}.

Before performing the final split, one needs to check that the pure gauge subspace lies in the constraint subspace. Again, one could check this by lengthy calculation, as was done in \cite[Thm~4.1]{Moncriefdecomp} for the $\Lambda=0$ case; instead, we appeal to the geometrical method of \cite[Prop.~3.2]{FM} whose result is  unaffected by the inclusion of a cosmological constant.

With the two preceding results and, as argued earlier, ellipticity of $D\Phi(\bh,\bvarpi)\circ D\Phi^{*}(\bh,\bvarpi)$ unaffected by addition of a cosmological constant, the subspace $\ker D\Phi(\bh,\bvarpi)$ can be decomposed, see the argument in between Theorem~{4.1} and Theorem~{4.2} of \cite{Moncriefdecomp}, into
\begin{equation}
\ker D\Phi(\bh,\bvarpi) = \range (U \circ D\Phi(\bh,\bvarpi)^{*})  \oplus \ker ((U \circ D\Phi(\bh,\bvarpi)^{*})^{*} \cap \ker D \Phi(\bh,\bvarpi),
\end{equation}
where the first space is pure gauge and the second space is the physical space.

Therefore the final split of the initial data is
\begin{multline}
C^{\infty}(S^{0}_{2}(\Sigma)) \times C^{\infty}(\tilde{S}^{2}_{0}(\Sigma))=\range D\Phi(\bh,\bvarpi)^{*} \oplus \range (U \circ D\Phi(\bh,\bvarpi)^{*}) \\ \oplus \ker ((U \circ D\Phi(\bh,\bvarpi)^{*})^{*} \cap \ker D \Phi(\bh,\bvarpi),
\end{multline}
which is the same result as the $\Lambda=0$ case from \cite[Thm~4.2]{Moncriefdecomp}. This decomposition allows one to prove that on the space of initial data obeying the constraints, $\ker  D \Phi(\bh,\bvarpi)$, the only degeneracies of the ADM symplectic product are pure gauge. We now give the analogue of \cite[Prop.~4.38]{FM_ECS}.
\begin{thm}
\label{thm:symplecticorthogonalcomplement}
The ADM symplectic orthogonal complement to the subspace $\ker D\Phi(\bh,\bvarpi)$ is the pure gauge space $\range (U \circ D\Phi(\bh,\bvarpi)^{*})$.
\end{thm}
\begin{proof} Let $(\tilde{\bgamma}^{(3)},\tilde{\bp}) \in \ker D\Phi(\bh,\bvarpi)$ satisfy $\omega^{ADM}_{(\bh,\bvarpi)}((\bgamma^{(3)},\bp);(\tilde{\bgamma}^{(3)},\tilde{\bp}))=0$ for all $(\bgamma^{(3)},\bp) \in \ker D\Phi(\bh,\bvarpi)$. Then by \eqref{eqn:ADMsympUinner},
\begin{equation}
\langle (\bgamma^{(3)},\bp); U^{-1}(\tilde{\bgamma}^{(3)},\tilde{\bp})) \rangle=0
\end{equation}
and so $U^{-1}(\tilde{\bgamma}^{(3)},\tilde{\bp})$ is orthogonal to $\ker D\Phi(\bh,\bvarpi)$. By the first Moncrief split \eqref{eqn:firstmoncriefsplit} this means that $U^{-1}(\tilde{\bgamma}^{(3)},\tilde{\bp}) \in \range D\Phi(\bh,\bvarpi)^{*}$. Hence
\begin{equation}
(\tilde{\bgamma}^{(3)},\tilde{\bp}) \in \range(U \circ D\Phi(\bh,\bvarpi)^{*})
\end{equation}
and is therefore pure gauge.  \end{proof}

\subsection{Proof of Theorem~\ref{thm:radical}} 

The main issue is to translate Theorem~\ref{thm:symplecticorthogonalcomplement} into
the setting studied in the main body of the paper.
Begin by taking an arbitrary smooth spacelike Cauchy surface $\Sigma$ and denote by $\mathscr{N}$ a normal neighbourhood of $\Sigma$. (For details about normal neighbourhoods, see the second paragraph of section~\ref{sec:synchronous}.) Assume that a solution $\bgamma^{\prime} \in \mathscr{S}(M)$ is a degeneracy
of the symplectic form $\omega$, i.e., $\omega(\bgamma^{\prime},\bgamma)=0$  for all  $\bgamma \in \mathscr{S}(M)$.
Without loss of generality, $\bgamma^{\prime}$ may be chosen synchronous near $\Sigma$; it will
be enough to restrict attention to synchronous $\bgamma$ as well. (Theorem~\ref{thm:synchronous} means that we can gauge transform any solution to the synchronous gauge near $\Sigma$ and since, by Lemma~\ref{lem:puregaugedegen},  pure gauge is a degeneracy then $\bgamma^{\prime}$ will still be a degeneracy of $\omega$.) 

We now restrict our attention to the normal neighbourhood $\mathscr{N}$, on which we can introduce Gaussian normal coordinates. In such coordinates the spacetime metric takes the form $\bg=-dt \otimes dt + \tilde{h}_{ij}dx^{i}\otimes dx^{j}$ and the synchronous condition is precisely $\gamma_{0\mu}=0$. The solutions $\bgamma^{\prime},\bgamma$ correspond to solutions to the linearized ADM equations
about the background $(\mathscr{N},\bg|_{\mathscr{N}})$ in the slicing given by the Gaussian
normal coordinates: thus we have unit lapse, vanishing shift (and vanishing linearizations thereof). 
The corresponding ADM  Cauchy data are $(\bgamma^{\prime(3)},\bp^{\prime}),(\bgamma^{(3)},\bp)\in C^{\infty}(S^{0}_{2}(\Sigma))\times C^{\infty}(\tilde{S}^{2}_{0}(\Sigma))$ respectively, where in these coordinates
\begin{equation}
\label{eqn:gammaij}
\gamma^{(3)}_{ij}=\gamma_{ij}|_{\Sigma}
\end{equation}
\begin{multline}
\label{eqn:pij}
p^{ij}=\sqrt{h}
\frac{\gamma^{(3)}}{4}\left(h^{im}h^{jn}-h^{ij}h^{mn}\right)\partial_{0}h_{mn} \\
 -\frac{\sqrt{h}}{2}\left(\gamma^{im}_{(3)}h^{jn}+h^{im}\gamma^{jn}_{(3)}-\gamma^{ij}_{(3)}h^{mn}-h^{ij}\gamma^{mn}_{(3)}\right) \partial_{0}h_{mn} \\ 
+ \frac{\sqrt{h}}{2}\left(h^{im}h^{jn}-h^{ij}h^{mn}\right)\left((\nabla_{n}\gamma)_{mn}|_{\Sigma}+\frac{1}{2}h^{kl}\left(\partial_{0}h_{ml}\gamma^{(3)}_{kn}+\partial_{0}h_{nl}\gamma^{(3)}_{mk}\right)\right).
\end{multline}
Using, for convenience, Gaussian normal coordinates one may show\footnote{We 
caution that the relationship between $\bpi$ and $\bp$ is not straightforward, although they do coincide 
on constant time hypersurfaces in Minkowski space.} that
\begin{equation}
\label{eqn:omegaADMlink}
\omega(\bgamma^{\prime},\bgamma)=\omega^{ADM}_{(\bh,\bvarpi)}((\bgamma^{\prime(3)},\bp^{\prime});(\bgamma^{(3)}\bp)).
\end{equation}
By Theorem~\ref{thm:symplecticorthogonalcomplement} and degeneracy of $\bgamma'$, $(\bgamma^{\prime(3)},\bp^{\prime})$ is data for a pure gauge solution; therefore on the region $\mathscr{N}$, $\bgamma^{\prime}=\pounds_{\bw}\bg$ for some $\bw \in C^{\infty}(T^{1}_{0}(\mathscr{N}))$. Now perform a global gauge transformation on $\bgamma^{\prime}$ using a vector field $\bv\in C^{\infty}(T^{1}_{0}(M))$ satisfying $\bv=-\bw$ on an open neighbourhood of $\Sigma$ within $\mathscr{N}$. The result will still be both a solution and a degeneracy in $\mathscr{S}(M)$ but has $\Data_{\Sigma}(\bgamma^{\prime}-\pounds_{\bv}\bg)=(0,0)$ and therefore by Theorem~\ref{thm:uniqueness}, $\bgamma^{\prime} = \pounds_{\bu}\bg$ for some $\bu \in C^{\infty}(T^{1}_{0}(M))$. Note that due to the compactness of $\Sigma$, all three vector fields $\bw$, $\bv$ and $\bu$ will be spacelike-compact and hence so will their associated pure gauge perturbation. Therefore $\bgamma^{\prime}\in\mathscr{G}(M)=\hat{\mathscr{G}}(M)$. 

\section{Identity connecting the symplectic product and linearized constraints}
\label{sec:symprodconstraint}

\begin{thm} \label{thm:C}
For any $\bgamma\in{\mathscr T}(M)$, $\bw\in C^{\infty}(T^{1}_{0}(M))$
and smooth spacelike Cauchy surface $\Sigma$ with unit future-pointing normal vector $\bn$, we have
\begin{equation}
\omega_{\Sigma}(\bgamma,\pounds_{\bw}\bg) = 2\int_{\Sigma}{w^{a}L_{ab}(\bgamma)n^{b}dvol_{\bh}} = 2\int_{\Sigma}{w^{b} C^{\Sigma}_{b}(\Data_{\Sigma}(\bgamma))dvol_{\bh}}.
\end{equation}
\end{thm}
\begin{proof} Choose a vector field $\bv \in C^{\infty}(T^{1}_{0}(M))$ such that $\bv=\bw$ in a neighbourhood of $\Sigma$ and $\bv$ vanishes to the far past of $\Sigma$. (This trick originates, 
as far as we know, from \cite{Friedman}, see the paragraph preceding equation (79) in that reference). 
By \eqref{eqn:IVCbC} we have
\begin{equation}
\label{eqn:intCLn}
\int_{\Sigma}{w^{a}C^{\Sigma}_{a}(\Data_{\Sigma}(\bgamma))dvol_{\bh}}=\int_{\Sigma}{w^{a}L_{ab}(\bgamma)n^{b}dvol_{\bh}}.
\end{equation}
The proof now uses two identities  to successively re-express the right-hand side of \eqref{eqn:intCLn}.
\begin{lem} With $\bgamma$, $\bw$ and $\bv$ as above,
\label{lem:gaugelinE}
\begin{equation}
\label{eqn:gaugevLab}
\int_{\Sigma}{w^{a}L_{ab}(\bgamma)n^{b}dvol_{\bh}} = -\int_{M^{-}}{\nabla_{(a}v_{b)}L^{ab}(\bgamma)dvol_{\bg}},
\end{equation}
where $M^{-}=I^-(\Sigma)$ is the region to the past of the Cauchy surface $\Sigma$.
\end{lem}
\begin{proof} Using the properties of $\bv$ and the Gauss Theorem on the region $M^{-}$ we have
\begin{equation}
\int_{\Sigma}{w^{a}L_{ab}(\bgamma)n^{b}dvol_{\bh}} =- \int_{M^{-}}{\nabla^{b}(v^{a}L_{ab}(\bgamma))dvol_{\bg}},
\end{equation}
where we have used that $\bv$ vanishes to the far past. The right-hand side can be rearranged using the Leibniz rule, symmetry of $L_{ab}$ and $\nabla^{a}L_{ab}=0$ to give the result.
\end{proof}

Now utilise the pre-symplectic product \eqref{eqn:presymp} to re-express the right-hand side of \eqref{eqn:gaugevLab}.
\begin{lem} With $\bgamma$, $\bw$ and $\bv$ as above,
\label{lem:gaugesymprodlinE}
\begin{equation}
\omega_{\Sigma}(\bgamma,\pounds_{\bw}\bg) =- \int_{M^{-}}{2 \nabla_{(a}v_{b)}L^{ab}(\gamma)dvol_{\bg}}.
\end{equation}
\end{lem}
\begin{proof} Expanding the left-hand side and using~\eqref{eqn:presymp} and the properties of $\bv$ gives
\begin{equation}
\omega_{\Sigma}(\bgamma,\pounds_{\bw}\bg) = \int_{\Sigma}{n_{a}[2\nabla_{(b} v_{c)}\Pi^{abc}(\bgamma)-\gamma_{bc}\Pi^{abc}(\pounds_{\bv}\bg)]dvol_{\bh}}.
\end{equation}
Applying the Gauss Theorem on the region $M^{-}$ gives
\begin{equation}
\omega_{\Sigma}(\bgamma,\pounds_{\bw}\bg)  =-\int_{M^{-}}{\nabla_{a}[2\nabla_{(b} v_{c)}\Pi^{abc}(\bgamma)-\gamma_{bc}\Pi^{abc}(\pounds_{\bv}\bg)]dvol_{\bg}}.
\end{equation}
The integrand is
\begin{multline}
\label{eqn:stokesgaugeexpansion}
\nabla_{a}(2\nabla_{(b}v_{c)}\Pi^{abc}(\bgamma)-\gamma_{bc}\Pi^{abc}(\pounds_{\bv}\bg)) =2\nabla_{a}\nabla_{(b}v_{c)}\Pi^{abc}(\bgamma) \\+2\nabla_{(b}v_{c)}\nabla_{a}\Pi^{abc}(\bgamma)-\nabla_{a}\gamma_{bc}\Pi^{abc}(\pounds_{\bv}\bg)-\gamma_{bc}\nabla_{a}\Pi^{abc}(\pounds_{\bv}\bg).
\end{multline}
Using \eqref{eqn:momentum} and the symmetries of $T^{abcdef}$, the first and third terms cancel. The remaining two terms reduce to
\begin{equation}
2\nabla_{(b}v_{c)}\nabla_{a}\Pi^{abc}(\bgamma)-\gamma_{bc}\nabla_{a}\Pi^{abc}(\pounds_{\bv}\bg) = 2 \nabla_{(b}v_{c)}L^{bc}(\bgamma),
\end{equation}
where we used the first identity in \eqref{eqn:Eulerlagrange} and $L^{bc}(\pounds_{\bv}\bg)=0$. Hence we achieve the desired result. \end{proof}

The proof of Theorem~\ref{thm:C} is completed by combining the results of Lemma~\ref{lem:gaugelinE} and Lemma~\ref{lem:gaugesymprodlinE} with \eqref{eqn:intCLn}. \global\logotrue 
\end{proof}



\small

\begin{thebibliography}{10}
\setlength{\itemsep}{-1.5mm}
\providecommand{\urlprefix}{}
\expandafter\ifx\csname urlstyle\endcsname\relax
  \providecommand{\doi}[1]{doi:\discretionary{}{}{}#1}\else
  \providecommand{\doi}{doi:\discretionary{}{}{}\begingroup
  \urlstyle{rm}\Url}\fi

\bibitem{AMR}
R.~Abraham, J.~E. Marsden, and T.~Ratiu, \emph{Manifolds, tensor analysis, and
  applications}, \emph{Applied Mathematical Sciences}, Vol.~75, 2nd edn.
  (Springer-Verlag, New York, 1988).

\bibitem{Ottewill}
B.~Allen, A.~Folacci, and A.~C. Ottewill, Renormalized graviton stress-energy
  tensor in curved vacuum space-times, \emph{Phys. Rev. D} \textbf{38} (1988)
  1069--1082.

\bibitem{ADM}
R.~Arnowitt, S.~Deser, and C.~W. Misner, The dynamics of general relativity, in
  \emph{Gravitation: {A}n introduction to current research} (Wiley, New York,
  1962), pp. 227--265.

\bibitem{AshtekarMagnon}
A.~Ashtekar and A.~Magnon-Ashtekar, On the symplectic structure of general
  relativity, \emph{Commun. Math. Phys.} \textbf{86} (1982) 55--68.

\bibitem{QFTCSTBar}
C.~B{\"a}r and K.~Fredenhagen (eds.), \emph{Quantum field theory on curved
  spacetimes: Concepts and mathematical foundations}, \emph{Lecture Notes in
  Physics}, Vol. 786 (Springer-Verlag, Berlin, 2009).

\bibitem{Bar}
C.~B{\"a}r, N.~Ginoux, and F.~Pf{\"a}ffle, \emph{Wave equations on {L}orentzian
  manifolds and quantization}, ESI Lectures in Mathematics and Physics
  (European Mathematical Society (EMS), Z\"urich, 2007).

\bibitem{BergerEbin}
M.~Berger and D.~Ebin, Some decompositions of the space of symmetric tensors on
  a {R}iemannian manifold, \emph{J. Differential Geometry} \textbf{3} (1969)
  379--392.

\bibitem{BernalSanchezcausal}
A.~N. Bernal and M.~S{\'a}nchez, Globally hyperbolic spacetimes can be defined
  as `causal' instead of `strongly causal', \emph{Class. Quantum Grav.}
  \textbf{24} (2007) 745--749.

\bibitem{BiHoRi2004}
E.~Binz, R.~Honegger, and A.~Rieckers, Construction and uniqueness of the
  {$C^*$}-{W}eyl algebra over a general pre-symplectic space, \emph{J. Math.
  Phys.} \textbf{45} (2004) 2885--2907.

\bibitem{BracciStrocchi}
L.~Bracci and F.~Strocchi, Einstein's equations and locality, \emph{Commun.
  Math. Phys.} \textbf{24} (1972) 289--302.

\bibitem{BFV}
R.~Brunetti, K.~Fredenhagen, and R.~Verch, The generally covariant locality
  principle: A new paradigm for local quantum physics, \emph{Commun. Math.
  Phys.} \textbf{237} (2003) 31--68.

\bibitem{DeserWaldronII}
S.~Deser and A.~Waldron, Gauge invariances and phases of massive higher spins
  in (anti-) de {S}itter space, \emph{Phys. Rev. Lett.} \textbf{87} (2001)
  031601, 4.

\bibitem{DeserWaldronI}
S.~Deser and A.~Waldron, Partial masslessness of higher spins in ({A})d{S},
  \emph{Nuclear Phys. B} \textbf{607} (2001) 577--604.

\bibitem{DeserWaldronIII}
S.~Deser and A.~Waldron, Stability of massive cosmological gravitons,
  \emph{Phys. Lett. B} \textbf{508} (2001) 347--353.

\bibitem{DimockScalar}
J.~Dimock, Algebras of local observables on a manifold, \emph{Commun. Math.
  Phys.} \textbf{77} (1980) 219--228.

\bibitem{DimockEM}
J.~Dimock, Quantized electromagnetic field on a manifold, \emph{Rev. Math.
  Phys.} \textbf{4} (1992) 223--233.

\bibitem{CJFMJP}
C.~J. Fewster and M.~J. Pfenning, A quantum weak energy inequality for spin-one
  fields in curved space-time, \emph{J. Math. Phys.} \textbf{44} (2003)
  4480--4513.

\bibitem{FewVer:dynloc_theory}
C.~J. Fewster and R.~Verch, Dynamical locality and covariance: What makes a
  physical theory the same in all spacetimes?, \emph{{A}nnales
  H.~{P}oincar{\'e}} \textbf{13} (2012) 1613--1674.

\bibitem{FewVer:dynloc2}
C.~J. Fewster and R.~Verch, Dynamical locality of the free scalar field,
  \emph{{A}nnales H.~{P}oincar{\'e}} \textbf{13} (2012) 1675--1709.

\bibitem{FierzPauli}
M.~Fierz and W.~Pauli, On relativistic wave equations for particles of
  arbitrary spin in an electromagnetic field, \emph{Proc. Roy. Soc. (London)
  Ser. A.} \textbf{173} (1939) 211--232.

\bibitem{FM_ECS}
A.~E. {Fischer} and J.~E. {Marsden}, {The initial value problem and the
  dynamical formulation of general relativity}, in \emph{General Relativity: An
  Einstein centenary survey}, ed. {S.~W.~Hawking \& W.~Israel} (1979), pp.
  138--211.

\bibitem{FM}
A.~E. {Fischer} and J.~E. {Marsden}, {Topics in the dynamics of general
  relativity}, in \emph{Isolated Gravitating Systems in General Relativity},
  ed. {J.~Ehlers} (1979), pp. 322--395.

\bibitem{FordParkerInfrared}
L.~H. Ford and L.~Parker, Infrared divergences in a class of
  {R}obertson-{W}alker universes, \emph{Phys. Rev. D} \textbf{16} (1977)
  245--250.

\bibitem{FordParker}
L.~H. Ford and L.~Parker, Quantized gravitational wave perturbations in
  {R}obertson-{W}alker universes, \emph{Phys. Rev. D} \textbf{16} (1977)
  1601--1608.

\bibitem{FredenhagenRejzner}
K.~Fredenhagen and K.~Rejzner, Batalin-{V}ilkovisky formalism in the functional
  approach to classical field theory, \emph{Commun. Math. Phys.} \textbf{314}
  (2012) 93--127.

\bibitem{Friedlander}
F.~G. Friedlander, \emph{The wave equation on a curved space-time} (Cambridge
  University Press, Cambridge, 1975), {C}ambridge Monographs on Mathematical
  Physics, No. 2.

\bibitem{Friedman}
J.~L. Friedman, Generic instability of rotating relativistic stars,
  \emph{Commun. Math. Phys.} \textbf{62} (1978) 247--278.

\bibitem{Furlani:1999}
E.~P. Furlani, Quantization of massive vector fields in curved space-time,
  \emph{J. Math. Phys.} \textbf{40} (1999) 2611--2626.

\bibitem{Geroch}
R.~Geroch, Spinor structure of space-times in general relativity. {I}, \emph{J.
  Math. Phys.} \textbf{9} (1968) 1739--1744.

\bibitem{HackSchenkel2012}
T.-P. {Hack} and A.~{Schenkel}, {Linear bosonic and fermionic quantum gauge
  theories on curved spacetimes}, arXiv:1205.3484.

\bibitem{HawkingEllis}
S.~W. Hawking and G.~F.~R. Ellis, \emph{The large scale structure of
  space-time} (Cambridge University Press, London, 1973), {C}ambridge
  {M}onographs on {M}athematical {P}hysics, No. 1.

\bibitem{AHSKlargedist}
A.~Higuchi and S.~S. Kouris, Large-distance behaviour of the graviton two-point
  function in de {S}itter spacetime, \emph{Class. Quantum Grav.} \textbf{17}
  (2000) 3077--3090.

\bibitem{HiguchiMarolf}
A.~Higuchi, D.~Marolf, and I.~A. Morrison, {de Sitter invariance of the dS
  graviton vacuum}, \emph{Class. Quantum Grav.} \textbf{28} (2011) 245012.

\bibitem{Globalanalysis}
A.~Kriegl and P.~W. Michor, \emph{The convenient setting of global analysis},
  \emph{Mathematical Surveys and Monographs}, Vol.~53 (American Mathematical
  Society, Providence, RI, 1997).

\bibitem{LeeWald}
J.~Lee and R.~M. Wald, Local symmetries and constraints, \emph{J. Math. Phys.}
  \textbf{31} (1990) 725--743.

\bibitem{Lichnerowicz}
A.~Lichnerowicz, Propagateurs et commutateurs en relativit\'e g\'en\'erale,
  \emph{Inst. Hautes \'Etudes Sci. Publ. Math.}  (1961) 56.

\bibitem{ManVer1968}
J.~Manuceau and A.~Verbeure, Quasi-free states of the {${\rm C.C.R.}$}-algebra
  and {B}ogoliubov transformations, \emph{Commun. Math. Phys.} \textbf{9}
  (1968) 293--302.

\bibitem{MiaoWoodward}
S.~P. Miao, N.~C. Tsamis, and R.~P. Woodard, {Gauging away Physics},
  \emph{Class. Quantum Grav.} \textbf{28} (2011) 245013.

\bibitem{Moncriefdecomp}
V.~Moncrief, Decompositions of gravitational perturbations, \emph{J. Math.
  Phys.} \textbf{16} (1975) 1556--1560.

\bibitem{MoncrieflinstabI}
V.~Moncrief, Spacetime symmetries and linearization stability of the {E}instein
  equations. {I}, \emph{J. Math. Phys.} \textbf{16} (1975) 493--498.

\bibitem{MoncriefQuantInstab}
V.~Moncrief, Quantum linearization instabilities, \emph{General Relativity and
  Gravitation} \textbf{10} (1979) 93--97.

\bibitem{ONeill}
B.~O'Neill, \emph{Semi-{R}iemannian geometry}, \emph{Pure and Applied
  Mathematics}, Vol. 103 (Academic Press Inc. [Harcourt Brace Jovanovich
  Publishers], New York, 1983).

\bibitem{MJP}
M.~J. Pfenning, Quantization of the {M}axwell field in curved spacetimes of
  arbitrary dimension, \emph{Class. Quantum Grav.} \textbf{26} (2009) 135017.

\bibitem{Sanders_compact}
K.~{Sanders}, A note on spacelike and timelike compactness, arXiv:1211.2469.

\bibitem{Stewart}
J.~M. Stewart and M.~Walker, Perturbations of space-times in general
  relativity, \emph{Proc. Roy. Soc. (London) Ser. A} \textbf{341} (1974)
  49--74.

\bibitem{StrocchiII}
F.~Strocchi, Gauge problem in quantum field theory. {II}. {D}ifficulties of
  combining {E}instein equations and {W}ightman theory, \emph{Phys. Rev.}
  \textbf{166} (1968) 1302--1307.

\bibitem{Taylor}
M.~E. Taylor, \emph{Partial differential equations {I}. {B}asic theory},
  \emph{Applied Mathematical Sciences}, Vol. 115, 2nd edn. (Springer, New York,
  2011).

\bibitem{Wald}
R.~M. Wald, \emph{General relativity} (University of Chicago Press, Chicago,
  IL, 1984).

\bibitem{WaldQFT}
R.~M. Wald, \emph{Quantum field theory in curved spacetime and black hole
  thermodynamics}, Chicago Lectures in Physics (University of Chicago Press,
  Chicago, IL, 1994).

\bibitem{Warner}
F.~W. Warner, \emph{Foundations of differentiable manifolds and {L}ie groups},
  \emph{Graduate Texts in Mathematics}, Vol.~94 (Springer-Verlag, New York,
  1983).

\bibitem{Weinberg}
S.~Weinberg, \emph{Cosmology} (Oxford University Press, Oxford, 2008).

\end{thebibliography}

\end{document}